\documentclass[a4paper,12pt]{article}
\usepackage[english]{babel}
\usepackage[utf8]{inputenc}
\usepackage{amsmath, amsfonts, mathrsfs, latexsym, multicol} 
\usepackage{srcltx}
\usepackage{pgf,epsfig,multicol,graphics,pdfsync,url}
\usepackage[amsmath,thmmarks,thref,hyperref]{ntheorem}
\usepackage{units}
\usepackage{enumerate}

\numberwithin{equation}{section}

\theoremstyle{plain}
\newtheorem{theorem}{Theorem}[section]
\newtheorem{lemma}[theorem]{Lemma}
\newtheorem{proposition}[theorem]{Proposition}
\newtheorem{corollary}[theorem]{Corollary}
\newtheorem{assumption}[theorem]{Assumption}
\theorembodyfont{\upshape}

\theoremstyle{nonumberplain}

\theoremsymbol{\ensuremath{_\Box}}
\newtheorem{proof}{Proof}

\newcommand{\epsi}{\varepsilon}
\newcommand{\E}{{\mathrm{e}}}
\newcommand{\I}{\mathrm{i}}

\newcommand{\Or}{{\mathcal{O}}}
\newcommand{\R}{{\mathbb{R}}}
\newcommand{\C}{{\mathbb{C}}}
\newcommand{\N}{{\mathbb{N}}}

\newcommand{\D}{{\mathrm{d}}}
\newcommand{\Hi}{ \mathcal{H} }
\newcommand{\B}{{\mathcal{B}}}
\newcommand{\trace}{\mathrm{Tr}}

\providecommand{\tracespin}{\trace_{\C^2}}

\providecommand{\pstrans}{\R^{2d}}

\providecommand{\Op}{\mathrm{Op}}
\providecommand{\Optrans}{\Op_{\pstrans}}
\providecommand{\Opspin}{\Op_{\Stwo}}
\providecommand{\Opsum}{\Op_{\Sigma}}
\providecommand{\projs}{\mathcal{P}}

\newcommand{\Stwo}{{\mathbb{S}^2}}
\providecommand{\Weyl}{\sharp}

\providecommand{\spinW}{\Weyl_{\Stwo}}

\providecommand{\Wigner}{\mathcal{W}}
\providecommand{\Wignersum}{\Wigner_{\Sigma}}
\providecommand{\Wignertrans}{\Wigner_{\pstrans}}
\providecommand{\Wignerspin}{\Wigner_{\Stwo}}
\providecommand{\Cont}{\mathcal{C}}
\providecommand{\Hil}{\mathcal{H}}

\providecommand{\R}{\mathbb{R}}

\providecommand{\C}{\mathbb{C}}
\renewcommand{\C}{\mathbb{C}}

\providecommand{\N}{\mathbb{N}}

\providecommand{\eps}{\varepsilon}

\providecommand{\abs}[1]{\left \lvert #1 \right \rvert}
\providecommand{\sabs}[1]{\lvert #1 \vert}
\providecommand{\babs}[1]{\bigl \lvert #1 \bigr \rvert}

\providecommand{\norm}[1]{\left \lVert #1 \right \rVert}
\providecommand{\snorm}[1]{\lVert #1 \rVert}
\providecommand{\bnorm}[1]{\bigl \lVert #1 \bigr \rVert}
\providecommand{\Bnorm}[1]{\Bigl \lVert #1 \Bigr \rVert}

\providecommand{\sexpval}[1]{\langle #1 \rangle}

\providecommand{\trace}{\mathrm{Tr} \,}

\providecommand{\dd}{\mathrm{d}}
\providecommand{\id}{\mathbf{1}}

\providecommand{\order}{\mathcal{O}}

\providecommand{\e}{\mathrm{e}}
\providecommand{\ii}{\mathrm{i}}

\providecommand{\ie}{i.~e.~}
\providecommand{\eg}{e.~g.~}
\providecommand{\cf}{cf.~}

\newcommand{\rom}{\renewcommand{\labelenumi}{{\rm(\roman{enumi})}}}
\newcommand{\brom}{\begin{enumerate}\rom}
\newcommand{\erom}{\end{enumerate}}

\title{Semiclassics for particles with spin via a Wigner-Weyl-type calculus\footnote{This work was supported by the German-Israeli Foundation under Grant no. 980/2007.}}
\author{Omri Gat${}^{\star}$, Max Lein${}^{\dagger}$ \& Stefan Teufel${}^{\ddagger}$}

\date{\today}

\begin{document}

\maketitle
\thispagestyle{empty}

\vspace{-9mm}
\begin{center}
	$^{\star}$ Hebrew University of Jerusalem, Racah institute of Physics, \linebreak
	Jerusalem 91904, Israel. \linebreak
	{\footnotesize \texttt{omrigat@cc.huji.ac.il}}
	\\[2mm]
	$^{\dagger}$ Kyushu University, Faculty of Mathematics \linebreak
	744 Motooka, Nishi-ku, Fukuoka-shi, Fukuoka-ken, Japan. \linebreak 
	{\footnotesize \texttt{lein@ma.tum.de}}
	\\[2mm]
	$^{\ddagger}$ Eberhard Karls Universit\"at T\"ubingen, Mathematisches Institut \linebreak
	Auf der Morgenstelle 10, 72076 T\"ubingen, Germany. \linebreak
	{\footnotesize \texttt{stefan.teufel@uni-tuebingen.de}}
\end{center}
\begin{abstract}
	We show how to relate the full quantum dynamics of a spin-$\nicefrac{1}{2}$ particle on $\R^d$ to a classical Hamiltonian dynamics on the enlarged phase space $\pstrans \times \Stwo$ up to errors of second order in the semiclassical parameter. This is done via an Egorov-type theorem for normal Wigner-Weyl calculus for $\R^d$ \cite{Lein:quantization_semiclassics:2010,Folland:harmonic_analysis_hase_space:1989} combined with the Stratonovich-Weyl calculus for $SU(2)$ \cite{Varilly_Gracia_Bondia:Moyal_rep_spin:1989}. 
	For a specific class of Hamiltonians, including the Rabi- and Jaynes-Cummings model, we prove an Egorov theorem for times much longer than the semiclassical time scale.
	We illustrate the approach for a simple model of the Stern-Gerlach experiment.
\end{abstract}

\tableofcontents

\addtolength{\topmargin}{1.5cm}
\addtolength{\textheight}{-2cm}

\newpage

\section{Introduction}
We consider the semiclassical limit of a quantum system with spin. The underlying state space is the Hilbert space $\Hil = L^2(\R^d,\C^n)$ of square-integrable functions on configuration space $\R^d$ taking values in $\C^n$. The Hamiltonian $\hat{H}$ generating the time evolution on $\Hi$ is assumed to be the Weyl quantization of a matrix-valued symbol $H : \R^{2d} \to \B(\C^n)$,~\ie
\begin{align*}
	\hat{H} = H(\epsi, x,-\ii \eps \nabla_x) := \Op \bigl ( H(\eps) \bigr ) 
	\, .
\end{align*}
Under appropriate conditions on the function $H$ its Weyl quantization $\hat{H}$ is a self-adjoint operator on $\Hi$ and generates a unitary group $\E^{-\I \hat{H} \frac{t}{\epsi}}$. We will be interested in the semiclassical asymptotics of the evolution of semiclassical observables. Let $A:\R^{2d} \to \B(\C^n)$ be another matrix-valued function on classical phase space, then the time evolution of the semiclassical operator $\hat{A} = \Op(A)$
is given by
\begin{align*}
	\hat{A}(t) := \E^{+ \ii \hat{H} \frac{t}{\eps}} \, \hat{A} \, \E^{-\ii \hat{H} \frac{t}{\eps}} 
\end{align*}
and $\hat{A}(t)$ solves the Heisenberg equation of motion
\begin{align*}
	\frac{\D}{\D t} \hat{A}(t) = \frac{\ii}{\eps} \bigl [ \hat{H} , \hat{A}(t) \bigr ]
	. 
\end{align*}
Instead of solving the full quantum dynamics, one is often interested in simpler, approximate solutions in terms of a classical Hamiltonian flow. In the scalar case $n=1$ this is the standard semiclassical problem. Under appropriate technical conditions the Egorov Theorem states that 
\begin{align}
	\sup_{t\in[0,T]} \norm{\hat{A}(t) - \Op \bigl ( A \circ \Phi^0_t \bigr )} = \Or(\epsi)\,,
	\label{ego1}
\end{align}
where $\Phi^0_t : \R^{2d} \to \R^{2d}$ denotes the classical Hamiltonian flow corresponding to the principal symbol $H_0(q,p)$ of $H(q,p) = H_0(q,p) + \Or(\epsi)$. That is, on bounded time intervals one can approximate the quantum mechanical time evolution of a semiclassical observable by transporting its symbol along a Hamiltonian flow. As a corollary, the Egorov theorem also implies a semiclassical evolution for states: here, the Wigner transform of the time-evolved density operator is compared to the Wigner transform of the density operator at $t = 0$ transported in time using the classical flow. In addition to its physical and practlical relevance for understanding the quantum-time-evolution,
the Egorov theorem is also the basis for a number of further mathematical results connecting properties of quantum and classical systems.
For example, in \cite{Zelditch:1996tj,Bolte_Glaser:semiclassical_limit_matrix-valued_operators:2004}
quantum-ergodicity for classically ergodic systems was proved based on the Egorov theorem.

The result  \eqref{ego1} has been generalized in many directions. For example, when replacing the flow $\Phi^0_t$ by the flow $\Phi^{\eps}_t$ generated by $H_0(q,p) + \epsi H_1(q,p)$, where $H(q,p) = H_0(q,p) + \epsi H_1(q,p) + \Or(\epsi^2)$, then the error is of order $\eps^2$, 
\begin{align}
	\sup_{t\in[0,T]} \norm{\hat{A}(t) - \Op \bigl ( A \circ \Phi^{\eps}_t \bigr )} = \Or(\epsi^2)
	.
	\label{ego2}
\end{align}
It is also well known how to construct higher-order approximations to $\hat{A}(t) $. However, they cannot be written as the composition of $A$ with a flow map on phase space $\R^{2d}$, see \eg \cite[Chaptre~IV~{\S}6]{Robert:tour_semiclassique:1987}. Depending on the details of the classical flow, one can potentially extend this approximation to longer time scales \cite{Bambusi_Graffi_Paul:semiclassics_Ehrenfest:1999,Bouzouina_Robert:uniform_semiclassical_estimates_Ehrenfest:2002}. 
 
In this work we consider the case $n>1$ of matrix valued symbols. Then basically two distinct cases appear: If the principal symbol $H_0$ of $H$ has eigenspaces that are nontrivial functions on $\R^{2d}$, then this is an adiabatic problem. Hence, in a first step the total state space $\Hi = L^2(\R^d, \C^n)$ of the system can be decomposed into orthogonal subspaces $\Hi_j$ that are each unitarily equivalent to spaces of the form $L^2(\R^d, \C^{n_j})$ with $\sum_j n_j=n$, see \cite{LITTLEJOHN:1992vg,Emmrich:1996vv,Nenciu:2004ex,PST:sapt:2002,Stiepan_Teufel:semiclassics_op_valued_symbols:2012}. The operator $\hat{H}$ is then block-diagonal with respect to these subspaces up to errors of order $\epsi^\infty$ and thus the reduced problems on each $L^2(\R^d, \C^{n_j})$ can be analyzed independently. If a subspace $\Hi_j$ is related to an isolated eigenvalue band $E_j$ of the principal symbol $H_0$, \ie $E_j(q,p)$ depends smoothly on $(q,p)$ and is an eigenvalue of $H_0(q,p)$ with constant multiplicity, then the block of $\hat{H}$ on that subspace is unitarily equivalent to a semiclassical Hamiltonian $\hat{H}_j$ on $L^2(\R^d, \C^{n_j})$ with principal symbol $E_j \, {\bf 1}_{\C^{n_j}}$ where ${\bf 1}_{\C^{n_j}}= {\rm diag(1,\cdots,1)}$. 

These semiclassical Hamiltonians $H_j \simeq E_j \, \mathbf{1}_{\C^{n_j}} + \order(\eps)$ are an example for   the type of Hamiltonian we will study in this paper: here, the principal symbol $H_0$ is a scalar multiple of the identity matrix. Since higher order terms don't contribute on the time scales we are interested in, we assume for the following that the symbol of the Hamiltonian has the form
\begin{align*}
	H(q,p) = H_0(q,p) + \eps \, H_1(q,p)
	= h_0(q,p) \, {\bf 1}_{\C^n} + \eps \, H_1(q,p)
	\, . 
\end{align*}
Also the semiclassical limit for this problem has been studied extensively in the literature, see \eg \cite{Bolte_Glaser:semiclassical_limit_matrix-valued_operators:2004} and references therein. However, the fact that $A$ is matrix-valued and that the quantum dynamics is non-trivial on the spin degrees of freedom makes it impossible to approximate $\hat{A}(t)$ in terms of a classical flow only on $\R^{2d}$ as in (\ref{ego1}) or (\ref{ego2}). Instead it is shown in \cite{Bolte_Glaser:semiclassical_limit_matrix-valued_operators:2004} based on \cite{Bolte_Keppeler:semiclassics_Dirac:1999} that the symbol of $\hat{A}(t)$ can be approximated, to leading order, by
\begin{equation}\label{eqn:skew_product_flow}
	A(t,q,p) = D^*(t,q,p) \, \bigl ( A \circ \Phi^0_t \bigr )(q,p)\, D(t,q,p)
	,
\end{equation}
where $\Phi^0_t$ is the Hamiltonian flow of the scalar principal symbol $h_0(q,p)$ of $H$. The orthogonal matrices $D(t,q,p)$ are generated by the subprincipal symbol $H_1$ according~to
\begin{align*}
	\tfrac{\D}{\D t} D(t,q,p) = - \I \, \bigl ( H_1\circ \Phi^0_t \bigr )(q,p) \, D(t,q,p)\,,\quad D(0,q,p) = {\bf 1}_{\C^n} 
	. 
\end{align*}
From the point of view of our results this description has two shortcomings: it gives only a leading order description, \ie an error term as in (\ref{ego1}) and the time-evolved $A(t,q,p)$ is not given in terms of a flow on phase space. 
 
The latter shortcoming was solved in \cite{Bolte_Glaser_Keppeler:ergodicity_spinning_particle:2001}, where the authors observed that one can use a Weyl calculus for spin developed by Stratonovich \cite{Stratonovich:distributions_rep_space:1957} and further elaborated upon by Gracia-Bondìa and Vàrilly \cite{Varilly_Gracia_Bondia:Moyal_rep_spin:1989} to map matrix-valued functions $A$ on $\R^{2d}$ to scalar functions $a := {\rm Symb}_{\Stwo}(A)$ on the \emph{extended phase space} $\Sigma := \R^{2d}\times \Stwo$. Note that the two-sphere $\Stwo$ carries a natural symplectic form, namely the volume form. The result is that $A(t,q,p)$ as in (\ref{eqn:skew_product_flow}) can indeed be written in terms of a so-called skew-product flow $\Phi^{\rm skew}_t$ on $\R^{2d}\times \Stwo$,
\begin{align*}
	{\rm Symb}_{\Stwo} \bigl ( A(t,q,p) \bigr ) =: a(t,q,p,n) = \bigl ( a \circ \Phi^{\rm skew}_t \bigr ) (q,p,n)
	\, .
\end{align*}
Note that the skew-product flow $\Phi^{\rm skew}_t$ is not a Hamiltonian flow and that the initial ``spin'' $n$ has no effect on the dynamics of the translational degrees of freedom, \ie with $\Phi^{\rm skew}_t =: \bigl ( Q^{\rm skew}_t , P^{\rm skew}_t , N^{\rm skew}_t \bigr)$ it holds that $Q^{\rm skew}_t(q,p,n) = Q^{\rm skew}_t(q,p)$ and $P^{\rm skew}_t(q,p,n) = P^{\rm skew}_t(q,p)$ are both independent of $n$. This is somewhat unsatisfactory from a physical point of view, since the paradigmatic experiment measuring the spin of a particle, the Stern-Gerlach experiment, is based on the fact that the trajectory of a particle with spin in an inhomogeneous magnetic field depends on its spin orientation. However, this is a small effect not seen in the leading-order approximation.

The main new result of the present paper is the construction of a \emph{Hamiltonian} flow $\Phi^{\eps}_t$ on the extended phase space $\Sigma = \R^{2d} \times \Stwo$ that \emph{includes the influence of the spin on the translational degrees of freedom} and provides a better approximation to the evolution of scalar observables $A= a \, {\bf 1}_{\C^n}$,
\begin{align}
	\sup_{t\in[0,T]} \norm{\hat{A}(t) - \Opsum \bigl ( a \circ \Phi^{\eps}_t \bigr )} = \Or(\epsi^2) 
	. 
	\label{intro:eqn:O_1_time_Egorov}
\end{align}
Here $\Opsum$ now denotes the quantization map from functions on $\Sigma$ to operators on $L^2(\R^{2d},\C^n)$ and $a(q,p,n) := a(q,p)$ (\cf Section~\ref{Weyl_calculus} for details).
 
For a special class of Hamiltonians on $L^2(\R,\C^2)$, including the Rabi- and the Jaynes-Cummings models, we show that the above approximation remains valid for longer times: Assume $h_0(q,p) = \frac{1}{2} (p^2 + \omega^2 q^2)$ is the one-dimensional harmonic oscillator Hamiltonian and suppose that the entries of $H_1(q,p)$ polynomials in $q$ and $p$ of degree at most one. Then we can prove \eqref{intro:eqn:O_1_time_Egorov} for longer times at the expense of a larger error. More precisely, we arrive at the estimate 
\begin{align}
	\sup_{t\in[0,T/\epsi^\gamma ]} \norm{\hat{A}(t) - \Opsum \bigl ( a \circ \Phi^{\eps}_t \bigr )} = \Or(\epsi^{\nicefrac{3}{2}-3\gamma})
	\label{ego3}
\end{align}
for any $0 \leq \gamma < \nicefrac{1}{2}$. Note that on this time scale the influence of spin on the position of the particle can be of order $\epsi^{1-2\gamma}$, so almost of order $1$. Moreover, even for general initial observables $\hat{B} = \Opsum(b)$ with $b=b(q,p,n)$ depending also on spin, we can go to slightly longer times and show that
\begin{align}
	\sup_{t\in[0,T/\epsi^\gamma ]} \norm{\hat{B}(t) - \Opsum \bigl ( b \circ \Phi^{\eps}_t \bigr )} = \Or(\epsi^{1-4\gamma}) 
	\label{ego4}
\end{align}
holds for any $0 \leq \gamma < \nicefrac{1}{4}$. Concrete applications of these results to the Rabi and the Jaynes-Cummings model are discussed elsewhere \cite{Gat:2013vm}. Note that without spin the semiclassical approximation for the harmonic oscillator $h_0$ is exact. On the extended phase-space $\R^2\times \Stwo$ this is no longer true, even when $h_1$ is linear in $q$ and $p$. However, in Theorem~\ref{exact} we show, that there is still an exact evolution equation for the symbol $a(t)$ of $\hat{A}(t)$. But this equation can no longer be solved by classical transport.

The paper is organized as follows: In Section~\ref{Weyl_calculus} we first recall the standard Weyl calculus on $\R^{2d}$, then introduce the Stratonovich-Weyl calculus on $\Stwo$ and finally show how to combine them. While this idea is not new (\eg \cite{Bolte_Glaser_Keppeler:ergodicity_spinning_particle:2001,Teufel:adiabatic_perturbation_theory:2003}),   the results of Section~\ref{comb} have not been worked out in this form before. Section~\ref{Egorov} contains the rigorous statements and the proofs of the general and the long-time Egorov theorems. 
Their proofs rely on estimates on the derivatives of the classical flow on $\Sigma$ which we show in Section~\ref{flow_estimates}. In the final Section~\ref{Stern-Gerlach} we illustrate the method by applying it to a simple model for the Stern-Gerlach experiment. We conclude, in particular, that the semiclassical approximation captures the spin-dependent splitting of wave packets correctly. \smallskip

\section{Weyl calculus}
\label{Weyl_calculus}
\subsection{Standard Weyl calculus on $\R^{2d}$}%
We briefly recall the most important features of Weyl calculus and specific properties that we will use in the following. Readers familiar with the Weyl calculus and not interested in the mathematical technicalities can skip this section. For more details on $\epsi$-pseudodifferential operators we refer \eg to \cite{Kumanogo:pseudodiff:1981,Taylor:PsiDO:1981,Martinez:intro_microlocal_analysis:2002}. For a short summary on Weyl calculus with operator-valued symbols the readers can also consult Appendix~A of \cite{Teufel:adiabatic_perturbation_theory:2003}.

One common and well-studied class of pseudodifferential operators are those which are Weyl quantizations of symbols. A \emph{symbol $f$ of order $k \in \R$} is a smooth function from $\R^{2d}$ to $\mathcal{B}(\C^n)$ such that 
\begin{align*}
	\norm{f}_{k,r} := \max_{\tiny \substack{\alpha , \beta \in \N_0^d \\ \abs{\alpha} + \abs{\beta} \leq r}} \sup_{(q,p) \in \R^{2d}} \Bnorm{\sexpval{p}^{-k+\abs{\alpha}} \, \partial_p^{\alpha} \partial_q^{\beta} f(q,p)}_{\mathcal{B}(\C^n)} 
\end{align*}
is bounded for all $r \in \N_0$. Here $\sexpval{p}^k:=(1+\sabs{p}^2)^{\nicefrac{k}{2}}$. Equipped with the family of seminorms $\{ \norm{\cdot}_{k,r} \}_{r \in \N_0}$, the set of all symbols of order $k$, denoted with $S^k$, is a Fréchet space. Functions in $S^k$ have the property, that derivatives with respect to $p$ improve the decay with respect to $p$. In particular, any partial derivative of degree $k$ with respect to $p$ of a function in $S^k$ yields a bounded function with bounded derivatives. The space of uniformly bounded functions $A: [0,\eps_0) \longrightarrow S^k $ is denoted by $S^k(\eps)$.

For a Schwartz function $\psi \in \mathcal{S}(\R^d,\C^n)$ and $A\in S^k(\epsi )$ the action of the Weyl quantization $\hat{A} = \Optrans(A)$ of $A$ can be defined through the oscillatory integral 
\begin{align}
	(\hat{A} \psi)(x) = \bigl ( \Optrans(A) \psi \bigr )(x) = \frac{1}{(2\pi \eps)^d}\int_{\R^{2d}} \dd p \, \dd y \,A \bigl ( \eps , \tfrac{1}{2}(x+y) , p \bigr ) \, \e^{+ \frac{\ii}{\eps} p \cdot (x-y)} \, \psi (y) 
	. 
	\label{eq:wq}
\end{align}
For $k \leq 0$, by the Calderon-Vaillancourt theorem \cite[Théorème~II~36]{Robert:tour_semiclassique:1987}, $\hat{A}$ can be extended to a bounded operator on $\Hil = L^2(\R^{d},\C^n)$; we will discuss other boundedness criteria in Section~\ref{Egorov}. 

The inverse of $\Op$ is the Wigner transform: if $K_A$ is the operator kernel to $\hat{A} \in \mathcal{B} \bigl ( L^2(\R^d,\C^n) \bigr )$, then we define 
\begin{align*}
	\bigl ( \Wignertrans(\hat{A}) \bigr )(q,p) := \frac{1}{(2\pi)^{\nicefrac{d}{2}}} \int_{\R^d} \dd y \, \e^{- \ii y \cdot p} \, K_A \bigl ( q + \tfrac{\eps}{2} y , q - \tfrac{\eps}{2} y \bigr ) 
	. 
\end{align*}
The composition of operators induces a composition of symbols. For any $A \in S^{k_1}(\eps)$ and $B \in S^{k_2}(\eps)$, there exists a symbol $C\in S^{k_1+k_2}(\epsi )$ denoted by $C= A\# B$ such that $\hat{A} \, \hat{B}= \hat{C}$. The bilinear map $\#: S^{k_1} \times S^{k_2} \to S^{k_1+k_2}$ is called the Moyal product and it is continuous with respect to the Fréchet topologies uniformly in $\eps$, \ie for any $r\in\N_0$ there is a $\tilde{r} \in \N_0$ and a constant $c_{r,\tilde{r}}<\infty$ such that
\begin{align*}
	 \bnorm{(A\# B)(\eps)}_{k_1+k_2,r} \leq c_{r,\tilde{r}} \, \bnorm{A(\eps)}_{k_1,\tilde{r}} \, \bnorm{B(\eps)}_{k_2,\tilde{r}}
\end{align*}
holds true for all $\epsi\in[0,\epsi_0)$. The last statement follows \eg from inspecting the proof of Thm.~2.41 in \cite{Folland:harmonic_analysis_hase_space:1989}. If for $A\in S^k(\epsi )$ there exists a sequence $(A_j)_{j\in\N_0}$ in $S^k(\epsi )$ such that
\begin{align*}
	\sup_{\eps \in[0,\eps_0) } \Bnorm{\epsi^{-(m+1)} \Bigl ( A(\eps) - \sum_{j=0}^m \eps^j \, A_j(\eps)\Bigr )}_{k,r}  < \infty 
\end{align*}
for all $r \in \N_0 $ and $m\in\N_0$, then one writes $ A \asymp \sum_{j=0}^\infty \epsi^j A_j$ in $S^k(\epsi )$. If $A\in S^k(\epsi )$ has an asymptotic expansion with coefficients $A_j\in S^k $ not depending on $\epsi$, then $A$ is called a classical symbol, $A_0$ its principal symbol and $A_1$ its subprincipal symbol. The Moyal product $C:= A\# B \in S^{k_1+k_2}(\epsi )$ of symbols $A \in S^{k_1}(\epsi )$ and $B \in S^{k_2}(\epsi )$ has an explicit asymptotic expansion $ C \asymp \sum_{j=0}^\infty \eps^j \, C_j$ such that $C_j \in S^{k_1+k_2-j}(\eps)$ and the remainder maps
\begin{align*}
	R_{m+1}: S^{k_1} \times S^{k_2} \to S^{k_1+k_2-m-1} \,, \quad (A,B) \mapsto R_{m+1}:= \eps^{-(m+1)} \Bigl ( C(\eps) - \sum_{j=0}^m \eps^j \, C_j(\eps) \Bigr ) 
\end{align*}
are continuous. The expansion starts with the pointwise product $C_0 (\epsi)= A(\epsi)B(\epsi)$ and the Poisson bracket $C_1(\epsi) = - \tfrac{ \I}{2} \{ A (\epsi), B(\epsi) \}_{\R^{2d}}$, where
\begin{align*}
	\{ A , B \}_{\R^{2d}} := \nabla_p A \cdot \nabla_q B - \nabla_q A \cdot \nabla_p B 
	:= \sum_{j=1}^d \bigl ( \partial_{p_j} A \, \partial_{q_j} B - \partial_{q_j} A \, \partial_{p_j} B \bigr ) 
	.
\end{align*}
For classical symbols $A$ and $B$ the Moyal product $C:= A\# B$ is also a classical symbol with an asymptotic expansion starting with
\begin{align*}
	A \# B \;\asymp\; A_0 B_0 \;+\; \eps \bigl ( A_1 B_0 + A_0 B_1 - \tfrac{\I}{2} \{ A_0, B_0 \}_{\R^{2d}}
	\bigr ) + \order(\eps^2)
	\, .
\end{align*}
If $A = a \, \mathbf{1}_{\C^n}$ is a scalar multiple of the identity, then $A$ and all its derivatives commute pointwise with any $B$. As a consequence one can show that in this case
\begin{align}
	\bigl [ A_0 , B_0 \bigr ]_{\sharp} := A_0 \# B_0 - B_0\# A_0 \asymp -\I \eps \{ A_0 , B_0 \}_{\R^{2d}} + \Or(\epsi^3)
	.
	\label{weylcommu}
\end{align}
The fact that the remainder term in (\ref{weylcommu}) is of order $\epsi^3$ and not only $\epsi^2$ is at the basis of our higher order semiclassical approximations. It distinguishes Weyl quantization from other quantization rules.

\subsection{Stratonovich-Weyl calculus for spin} 
\label{quantizations:S2}

Similar to Weyl calculus that maps functions on phase space $\R^{2d}$ to operators on the Hilbert space $L^2(\R^{d})$, there is a Weyl calculus that associates functions on the compact phase space $\Stwo$ to operators on the finite-dimensional Hilbert spaces $\C^n$. It was first proposed by Stratonovich \cite{Stratonovich:distributions_rep_space:1957} and elaborated upon further by Gracia-Bondia and Varilly \cite{Varilly_Gracia_Bondia:Moyal_rep_spin:1989,Varilly_Gracia_Bondia_Schempp:Moyal_rep_qm:1990} and has been applied to study the dynamics at Josephson junctions \cite{Chuchem_Smith-Mannschott_Hiller_Kottos_Vardi_Cohen:quantum_dynamics_josephson_junction:2010}. 
 
To make the following as transparent as possible, we restrict the presentation to the case of $\C^2$  and use the letter $n$ to denote a point on the unit sphere $\Stwo$ from now on. Higher dimensional spin-spaces can be dealt with in complete analogy \cite{Varilly_Gracia_Bondia:Moyal_rep_spin:1989,Varilly_Gracia_Bondia_Schempp:Moyal_rep_qm:1990}. The basic observation is that any $2 \times 2$ matrix can be written as a linear combination of the identity matrix $\id_{\C^2}$ and the three Pauli matrices $\sigma_j$, $j = 1 , 2 , 3$, 
\begin{align}
	 {A} = \mathfrak{a}_0 \, \id_{\C^2} + \sum_{j = 1}^3 \mathfrak{a}_j \, \sigma_j 
	= \mathfrak{a}_0 \, \id_{\C^2} + \mathfrak{a} \cdot \sigma 
	\label{quantizations:eqn:matrix_sum_of_Pauli}
\end{align}
for some complex coefficients $\mathfrak{a}_0 , \ldots , \mathfrak{a}_3 \in \C$. Now the quantization map $\Opspin: C^\infty(\Stwo)\to \B(\C^2)$ can be defined most easily using the \emph{Stratonovich-Weyl kernel}
\begin{align}
	\Delta(n) := \tfrac{1}{2} \bigl ( \id_{\C^2} + \sqrt{3} \, n \cdot \sigma \bigr ) 
	, 
	\qquad 
	n \in \Stwo
	,
	\label{quantizations:eqn:Stratonovich-Weyl_kernel}
\end{align}
by setting
\begin{align}
	\Opspin(a) := \frac{1}{2\pi} \int_{\Stwo} \D n \, a(n) \, \Delta(n) \,.
	\label{Weyl_S2:eqn:quantization}
\end{align}
It is clearly many to one, but there is a natural way to define also a dequantization map $\Wignerspin: \B(\C^2)\to C^\infty(\Stwo)$ by
\begin{align}
	\bigl ( \Wignerspin( {A}) \bigr )(n) := \tracespin \bigl ( {A} \, \Delta(n) \bigr ) \,.
	\label{quantizations:eqn:defn_Wigner_trafo}
\end{align}
which maps onto the four-dimensional subspace $C_1(\Stwo):={\rm span}\{1,n_1,n_2,n_3\}$ of $C^\infty(\Stwo)$. Using $\tracespin \, \sigma_j = 0$, we find that the dequantization of $ {A}$ written as in \eqref{quantizations:eqn:matrix_sum_of_Pauli} is the linear polynomial
\begin{align}
	a(n) = \bigl ( \Wignerspin( {A}) \bigr )(n) = \mathfrak{a}_0 + \sqrt{3} \, \mathfrak{a} \cdot n 
	, 
	\label{quantizations:eqn:dequantization_linear_polynomial}
\end{align}
and $\Opspin:C_1(\Stwo)\to\B(\C^2) $ is indeed one-to-one with inverse $\Wignerspin$. The projection $\projs = \Wignerspin \circ \Opspin$ maps any $a\in C^\infty(\Stwo)$ to the representative in $C_1(\Stwo)$ that quantizes to the same matrix and  is explicitly given by
\begin{align}
	(\projs a)(n) &= \frac{1}{2\pi} \int_{\Stwo} \D k \, a(k) \, \tracespin \bigl ( \Delta(k) \, \Delta(n) \bigr ) 
	= \frac{1}{4\pi} \int_{\Stwo} \D k \, a(k) \, (1 + 3 k \cdot n) 
	\notag \\
	&= \frac{1}{4\pi} \int_{\Stwo} \D k \, a(k) + \sqrt{3} \, \left ( \frac{\sqrt{3}}{4\pi} \int_{\Stwo} \D k \, k \, a(k) \right ) \cdot n
	\label{quantizations:eqn:projs_definition} 
	\\
	&=: \mathfrak{a}_0 + \sqrt{3} \, \mathfrak{a} \cdot n 
	\notag 
\end{align}
For later reference we note that for all $a \in C^\infty(\Stwo)$ we have that
\begin{align}
	\projs a &= a && \mbox{for all $a\in C_1(\Stwo)$} 
	\label{C1prop} \\
	\Opspin \projs a &= \Opspin a&& \mbox{for all $a\in C^\infty(\Stwo)$. } 
	\label{projsprop}
\end{align}
Since the symbol $a$ of a given matrix $A$ is not unique, there is a priori also no unique way to define the corresponding Moyal product $\spinW$, \ie the matrix product on the level of functions on $\Stwo$. However, if we demand that it takes its values in $C_1(\Stwo)$, it is unique and one finds
\begin{align}
	(a \, \spinW \, b)(n) :& \negmedspace= \Wignerspin \bigl ( \Opspin(a) \, \Opspin(b) \bigr )
	= \Wignerspin \bigl ( \Opspin(\projs a) \, \Opspin(\projs b) \bigr ) 
	\notag \\
	&= \Wignerspin \bigl ( ( \mathfrak{a}_0 \, \id_{\C^2} + \mathfrak{a} \cdot \sigma ) \, ( \mathfrak{b}_0 \, \id_{\C^2} + \mathfrak{b} \cdot \sigma ) \bigr ) 
	\notag \\
	&= \Wignerspin \bigl ( ( \mathfrak{a}_0 \mathfrak{b}_0 + \mathfrak{a} \cdot \mathfrak{b} ) + ( \mathfrak{a}_0 \mathfrak{b} + \mathfrak{a} \mathfrak{b}_0 + \I \, \mathfrak{a} \wedge \mathfrak{b} ) \cdot \sigma \bigr ) 
	\notag \\
	&= ( \mathfrak{a}_0 \mathfrak{b}_0 + \mathfrak{a}\cdot \mathfrak{b} ) + \sqrt{3} \, ( \mathfrak{a}_0 \mathfrak{b} + \mathfrak{a} \mathfrak{b}_0 + \I \, \mathfrak{a} \wedge \mathfrak{b} ) \cdot n 
	.
	\label{SpinMoyal}
\end{align}
As for the calculus on $\R^{2d}$, the applicability to the semiclassical limit of the Heisenberg equations of motion rests on the observation that the commutator of operators corresponds to the Poisson bracket of symbols. The natural symplectic form on $\Stwo$ is the volume form $\eta$, which we normalize such that for two tangent vectors $v,w$ at $n\in \Stwo$
\[
\eta_n( v,w) = -\tfrac{\sqrt{3}}{2} (v\wedge w )\cdot n\,.
\]
Thus the Hamiltonian vector field associated to a function $a\in C^\infty(\Stwo)$ is 
\[
X^a = - \tfrac{2}{\sqrt{3}} \; \nabla_n a\wedge n
\]
and a short computation yields
\begin{align*}
	\{ a , b \}_{\Stwo} := \eta(X^a,X^b) = -\tfrac{2}{\sqrt{3}} \bigl ( \nabla_n a \wedge \nabla_n b \bigr ) \cdot n
	\,.
\end{align*}
Comparing with \eqref{SpinMoyal}, we immediately see that 
\begin{align}
	[ a , b ]_{\spinW} := a \, \spinW b - b \, \spinW a 
	= - \ii \, \bigl \{ \projs a , \projs b \bigr \}_{\Stwo} 
	. 
	\label{quantizations:eqn:spin_commutator}
\end{align}
Now there is an observation which is crucial for the following: 
\begin{lemma}
	Let $a,b \in C^\infty(\Stwo)$, then
	\begin{align}
		\bigl \{ \projs a , \projs b \bigr \}_{\Stwo} =\projs \bigl \{ \projs a , b \bigr \}_{\Stwo} 
		\label{lemmaPoisson}
		. 
	\end{align}
\end{lemma}
\begin{proof}
	We write $(\projs a)(n) = \mathfrak{a}_0 + \sqrt{3} \, \mathfrak{a} \cdot n $ and $(\projs b)(n) = \mathfrak{b}_0 + \sqrt{3} \, \mathfrak{b} \cdot n $. Then
\begin{align*}
	\projs \bigl \{ \projs a , b \bigr \}_{\Stwo} = - \frac{2}{4\pi} \int_{\Stwo} \D k \, \bigl ( \mathfrak{a} \wedge \nabla_n b(k) \bigr ) \cdot k - \sqrt{3} \, \left ( \frac{2\sqrt{3}}{4\pi} \int_{\Stwo} \D k \, k \, \bigl ( \mathfrak{a} \wedge \nabla_n b(k) \bigr ) \cdot k \right ) \cdot n\,.
\end{align*}
Integration by parts gives for the first term
\[
 \int_{\Stwo} \D k \, \bigl ( \mathfrak{a} \wedge \nabla_n b \bigr ) \cdot k = 
 \int_{\Stwo} \D k \, \epsilon_{jlm} \, \mathfrak{a}_j \, \partial_l b(k)\, k_m = 0
\]
and for the $j$th component of the second term
\begin{align*}
	\int_{\Stwo} \D k \, k_j \, \bigl ( \mathfrak{a} \wedge \nabla_n b(k) \bigr ) \cdot k 
	&= \int_{\Stwo} \D k \, k_j \, \epsilon_{lms} \, \mathfrak{a}_l \, \partial_m b(k)\, k_s 
	= - \int_{\Stwo} \D k \, \epsilon_{ljs} \, \mathfrak{a}_l \, b(k)\, k_s 
	\\
	&= - \epsilon_{jsl} \, \mathfrak{a}_l \, \int_{\Stwo} \D k \, b(k)\, k_s 
	= (\mathfrak{a}\wedge\mathfrak{b})_j \, \tfrac{4\pi}{\sqrt{3}}
	\, .
\end{align*}
Comparing with $\bigl \{ \projs a , \projs b \bigr \}_{\Stwo} = -2\sqrt{3}(\mathfrak{a}\wedge\mathfrak{b})\cdot n $ proves the claim.
\end{proof}

\begin{corollary}
Let $a \in C_1(\Stwo)$ and $b\in C^\infty(\Stwo)$. Then
\begin{equation}\label{spincom}
	\bigl [ \Opspin a, \Opspin b \bigr ] = - \I \, \Opspin \{ a , b \}_{\Stwo}\,.
\end{equation}
\end{corollary}
\begin{proof} 
	\begin{align*}
		\bigl [ \Opspin a, \Opspin b \bigr ] &= \Opspin \bigl [ a , b \bigr ]_{\spinW} 
		\stackrel{(\ref{quantizations:eqn:spin_commutator})}{=} - \I\, \Opspin \bigl \{ \projs a , \projs b \bigr \}_{\Stwo} 
		\stackrel{(\ref{lemmaPoisson})}{=} - \I \,\Opspin \projs\bigl \{ \projs a , b \bigr \}_{\Stwo}
		\\
		&\stackrel{(\ref{projsprop})}{=} - \I\, \Opspin \bigl \{ \projs a , b \bigr \}_{\Stwo} 
		\stackrel{(\ref{C1prop})}{=} - \I \,\Opspin \bigl \{ a , b \bigr \}_{\Stwo}\,.
	\end{align*}
\end{proof}

\subsection{Weyl calculus on $\R^{2d}\times \Stwo$}\label{comb}

It is now straightforward to define a Weyl calculus on the product $\Sigma = \R^{2d}\times \Stwo$. We say $a\in S^k_\Sigma(\epsi)$, if $a(\epsi)\in C^\infty(\Sigma)$ such that $\Opspin a\in S^k(\epsi)$. For $a\in S^k_\Sigma(\epsi)$ define $\Opsum$ and $\Wignersum$ by
\begin{align*}
	\Opsum(a) :& \negmedspace= \Optrans \bigl ( \Opspin(a) \bigr ) 
	, 
	\\
	\Wignersum(\hat{A}) :& \negmedspace= \Wignertrans \bigl ( \Wignerspin(\hat{A}) \bigr ) 
	, 
\end{align*}
and let 
\[
\{a,b\}_\Sigma := \{a,b\}_{\R^{2d}} + \tfrac{1}{\epsi} \{a,b\}_{\Stwo} \,.
\]
We now formulate our assumptions on the symbol of the Hamiltonian.
\begin{assumption}\label{assumption:hamiltonian_symbol}
	Let $h(q,p,n) = h_0(q,p) + \eps \, h_1(q,p,n)$ such that $h_0 \in S^2$ and 
	\[
	h_1(q,p,n) = \mathfrak{h}_0(q,p) + \sqrt{3} \, \mathfrak{h}(q,p)\cdot n
	\]
	with $ \mathfrak{h}_0$ and all components of $\mathfrak{h}(q,p)$ in $S^1$. 
\end{assumption}
Note that in view of \eqref{quantizations:eqn:projs_definition} and \eqref{projsprop}, assuming $h$ is a linear polynomial in $n$ is not a restriction, merely a convenient way to write the symbol. 
\begin{lemma}\label{lemma:commutator_to_Poisson}
	Let $h$ satisfy Assumption~\ref{assumption:hamiltonian_symbol} and let $a \in S^0_\Sigma(\epsi)$.
	\begin{enumerate}[(i)]
		\item Then 
		\begin{align*}
			\Bnorm{\tfrac{\I}{\epsi} \bigl [ \Opsum h \, , \Opsum a \bigr ] - \Opsum \{h,a\}_\Sigma}_{\B(\Hi)} = \Or(\epsi) 
			.
		\end{align*}
		\item If, in addition, there is a symbol $a_0$ such that $a(q,p,n) - a_0(q,p) = \order(\eps)$, then
		\begin{align*}
			\Bnorm{\tfrac{\I}{\epsi} \bigl [ \Opsum h \, , \Opsum a \bigr ] - \Opsum \{ h , a \}_{\Sigma}}_{\B(\Hi)} = \order(\eps^2)
			.
		\end{align*}
		\item Suppose that $h_0$ is a quadratic polynomial in the components of $p$ and $q$, and the components of $\mathfrak{h}_0$ and $\mathfrak{h}$ are linear polynomials in $q$ and $p$. Then 
		\begin{align*}
			\tfrac{\I}{\epsi} \bigl [ \Opsum h, \Opsum a \bigr ] &= \Opsum \{h,a\}_\Sigma
			\\
			&\qquad + \tfrac{\epsi}{2} \, \Optrans \big ( \{H_1, A\}_{\R^{2d}}
			-\{A,H_1\}_{\R^{2d}} - 2 \,\Opspin \{ h_1, a\}_{\R^{2d}}
			\big ) 
			\\
			&= \Opsum \{h,a\}_\Sigma - \tfrac{\epsi}{2}\,\Opsum\{ h_1, (1-\mathcal{P})a\}_{\R^{2d}}
		\end{align*}
		holds true where $H_1 = \Opspin h_1$ and $A=\Opspin a$.
	\end{enumerate}
\end{lemma}
\begin{proof} 
	\begin{enumerate}[(i)]
		\item For the scalar symbol $h_0\in S^2$ we have according to (\ref{weylcommu})
		\begin{align}
			\bigl [ \Opspin h_0 \, , \Opspin a \bigr ]_{\sharp} &= - \I \, \eps \, \bigl \{ \Opspin h_0 , \Opspin a \bigr \}_{\R^{2d}} + \Or(\eps^3)
			,
			\label{moy1}
		\end{align}
		where the remainder is order $\epsi^3$ in $S^{-1}(\epsi)$. For $h_1$ 
		\begin{align}
			\bigl [ \Opspin h_1 \, , \Opspin a \bigr ]_{\sharp} &= \bigl [ \Opspin h_1, \Opspin a \bigr ] + 
			\label{moy2} \\
			&\qquad 
			- \eps \, \tfrac{\ii}{2} \left ( \bigl \{ \Opspin h_1, \Opspin a \bigr \}_{\R^{2d}} - \bigl \{ \Opspin a, \Opspin h_1 \bigr \}_{\R^{2d}} \right ) + \Or(\epsi^2)
			\notag 
		\end{align}
		where the remainder is order $\epsi^2$ in $S^{-1}(\epsi)$. Hence
		\begin{align*}
			\tfrac{\I}{\epsi} \bigl [ \Opsum h \, , \Opsum a \bigr ] 
			&= \Optrans \Bigl ( \{\Opspin h_0, \Opspin a\}_{\R^{2d}} + \I [ \Opspin h_1, \Opspin a] 
			+ \Bigr . 
			\\
			&\qquad \qquad \Bigl . 
			+ \tfrac{\epsi}{2} \bigl ( \{\Opspin h_1, \Opspin a\}_{\R^{2d}}
			-\{\Opspin a, \Opspin h_1\}_{\R^{2d}} \bigr )
			\Bigr ) +\Or(\epsi^2)
			,
			\\
			&\stackrel{(\ref{spincom})}{=} \Optrans \Bigl ( \Opspin \{ h_0, a\}_{\R^{2d}}+ \Opspin \{h_1,a\}_{\Stwo} + \epsi \,\Opspin \{ h_1, a\}_{\R^{2d}}
			+ \Bigr . \\
			&\qquad \qquad \quad \Bigl . 
			+ \tfrac{\epsi}{2} \bigl ( \{\Opspin h_1, \Opspin a\}_{\R^{2d}}
			-\{\Opspin a, \Opspin h_1\}_{\R^{2d}}
			+ \Bigr . \\
			&\qquad \qquad \quad \Bigl . 
			-2 \,\Opspin \{ h_1, a\}_{\R^{2d}} \bigr )
			\Bigr ) + \Or(\epsi^2) 
			\\
			&= \Opsum \{h,a\}_\Sigma
			+ \tfrac{\epsi}{2} \, \Optrans \bigl ( \{\Opspin h_1, \Opspin a\}_{\R^{2d}}
			+ \Bigr . \\
			&\qquad \qquad \quad \Bigl . 
			-\{\Opspin a, \Opspin h_1\}_{\R^{2d}}
			-2 \,\Opspin \{ h_1, a\}_{\R^{2d}}\bigr ) + \Or(\epsi^2)
			,
		\end{align*}
		where the second term of order $\epsi$ is a bounded operator and the remainder is order $\epsi^2$ as a bounded operator. 
		\item Under the assumption of (ii), also the second term is of order $\epsi^2$, since 
		\begin{align*}
			\{\Opspin h_1, &\Opspin a\}_{\R^{2d}} - \{\Opspin a, \Opspin h_1\}_{\R^{2d}} - 2 \, \Opspin \{ h_1, a\}_{\R^{2d}} = 
			\\
			&= \{\Opspin h_1, \Opspin (a-a_0)\}_{\R^{2d}} - \{\Opspin (a-a_0), \Opspin h_1\}_{\R^{2d}} 
			+ \\
			&\quad \;
			- 2 \,\Opspin \{ h_1, (a-a_0)\}_{\R^{2d}} 
			\\
			&
			= \Or(\epsi) 
			.
		\end{align*}
		\item For the first equality in (iii) just note that there are no remainder terms in (\ref{moy1}) and (\ref{moy2}) in this case. 

		For symbols that are at most linear polynomials in $n$ the explicit remainder term vanishes.
		\begin{lemma}
		Let $a,b \in C^\infty(\Sigma) $ such that $a(q,p,\cdot), b(q,p,\cdot)\in C_1(\Stwo)$ for all $(q,p)\in\R^{2d}$. Then
		\[
		R(a,b):= \{\Opspin a, \Opspin b\}_{\R^{2d}}
		-\{\Opspin b, \Opspin a\}_{\R^{2d}}
		-2 \,\Opspin \{ a, b\}_{\R^{2d}}
		=0\,.
		\]
		\end{lemma}
		\begin{proof}
			Let $a = \mathfrak{a}_0 + \sqrt{3} \, \mathfrak{a} \cdot n$ and $b = \mathfrak{b}_0 + \sqrt{3} \, \mathfrak{b} \cdot n$ be the scalar symbols to the matrix-valued functions $A = \mathfrak{a}_0 \, \id_{\C^2} + \mathfrak{a} \cdot \sigma$ and $B = \mathfrak{b}_0 \, \id_{\C^2} + \mathfrak{b} \cdot \sigma$. We can write the difference of the Poisson brackets as the difference of two anti-commutators, 
			\begin{align}
				\{ A , B \}_{\R^d} - \{ B , A \}_{\R^d} &= \sum_{j = 1}^d \left ( [\partial_{p_j} A \, , \partial_{q_j} B]_+ - [\partial_{p_j} A \, , \partial_{q_j} B]_+ \right ) 
				. 
				\notag 
			\end{align}
			These anti-commutators can be expressed in terms of the coefficients, 
			\begin{align}
				[A \, , B]_+ &= \bigl ( \mathfrak{a}_0 \, \id_{\C^2} + \mathfrak{a} \cdot \sigma \bigr ) \, \bigl ( \mathfrak{b}_0 \, \id_{\C^2} + \mathfrak{b} \cdot \sigma \bigr ) - \bigl ( \mathfrak{b}_0 \, \id_{\C^2} + \mathfrak{b} \cdot \sigma \bigr ) \, \bigl ( \mathfrak{a}_0 \, \id_{\C^2} + \mathfrak{a} \cdot \sigma \bigr ) 
				\notag \\
				&= 2 \bigl ( \mathfrak{a}_0 \, \mathfrak{b}_0 + \mathfrak{a} \cdot \mathfrak{b} \bigr ) \, \id_{\C^2} + 2 \, \bigl ( \mathfrak{a}_0 \, \mathfrak{b} + \mathfrak{b}_0 \, \mathfrak{a} \bigr ) \cdot \sigma 
				. 
				\notag 
			\end{align}
			By a straightforward computation, we can verify that $\projs  \bigl ( 3 (\mathfrak{a} \cdot n) \, (\mathfrak{b} \cdot n ) \bigr ) = \mathfrak{a} \cdot \mathfrak{b}$ holds true and hence 
			\begin{align*}
				\Opspin  \bigl ( 3 (\mathfrak{a} \cdot n) \, (\mathfrak{b} \cdot n ) \bigr ) &= \Opspin \projs  \bigl ( 3 (\mathfrak{a} \cdot n) \, (\mathfrak{b} \cdot n ) \bigr ) 
				= (\mathfrak{a} \cdot \mathfrak{b}) \, \id_{\C^2} 
				. 
			\end{align*}
			This means the right-hand side of 
			\begin{align*}
				2 \, \Opspin \{ a , b \}_{\R^d} &= 2 \, \sum_{j = 1}^d \Opspin \Bigl ( \bigl ( \partial_{p_j}  \mathfrak{a}_0 + \sqrt{3} \, \partial_{p_j} \mathfrak{a} \cdot n \bigr ) \, \bigl ( \partial_{q_j}  \mathfrak{b}_0 + \sqrt{3} \, \partial_{q_j} \mathfrak{b} \cdot n \bigr ) \Bigr ) 
				+ \\
				&\qquad \qquad 
				+ \mbox{other terms}
				\\
				&= 2 \, \sum_{j = 1}^d \Opspin \Bigl ( \bigl ( \partial_{p_j} \mathfrak{a}_0 \, \partial_{q_j} \mathfrak{b}_0 + 3 \, (\partial_{p_j} \mathfrak{a} \cdot n) \, (\partial_{q_j} \mathfrak{b} \cdot n) \bigr ) 
				+ \Bigr . \\
				&\qquad \qquad \qquad \quad \Bigl . + 
				\sqrt{3} \, \bigl ( \partial_{q_j} \mathfrak{b}_0 \, \partial_{p_j} \mathfrak{a} + \partial_{p_j} \mathfrak{a}_0 \, \mathfrak{b} \bigr ) \cdot n \Bigr ) 
				+ \mbox{other terms}
				\\
				&= \sum_{j = 1}^d \bigl ( [\partial_{p_j} A \, , \partial_{q_j} B]_+ - [\partial_{q_j} A \, , \partial_{p_j} B]_+ \bigr ) 
			\end{align*}
			agrees with $\{ A , B \}_{\R^d} - \{ B , A \}_{\R^d}$, and $R(a,b)$ vanishes identically. 
		\end{proof}
		Since, according to \eqref{projsprop}, we can replace $\Opspin a= \Opspin \mathcal{P}a$, it follows that 
		\begin{align*}
			\{\Opspin h_1, &\Opspin a\}_{\R^{2d}} - \{\Opspin a, \Opspin h_1\}_{\R^{2d}} - 2 \, \Opspin \{ h_1, a\}_{\R^{2d}} 
			= \\
			&= R(h_1,a) - \Opspin \{ h_1, (1-\mathcal{P})a\}_{\R^{2d}} 
			= - \Opspin \{ h_1, (1-\mathcal{P})a\}_{\R^{2d}} 
			,
		\end{align*}
		which proves the second equality in (iii).
	\end{enumerate}
\end{proof}

\section{Egorov Theorems}
\label{Egorov}
Suppose $h$ satisfies Assumption~\ref{assumption:hamiltonian_symbol} and denote by $\Phi^{\eps}$ the flow associated to Hamilton's equation of motion 
\begin{align}
	\dot{q} &= + \nabla_p h 
	= + \nabla_p h_0 + \eps \, \bigl ( \nabla_p \mathfrak{h}_0 + \sqrt{3} \, \nabla_p (\mathfrak{h} \cdot n) \bigr ) 
	\label{Egorov:eqn:hamiltonian_eom}
	\\
	\dot{p} &= - \nabla_q h 
	= - \nabla_q h_0 - \eps \, \bigl ( \nabla_q \mathfrak{h}_0 + \sqrt{3} \, \nabla_q (\mathfrak{h} \cdot n) \bigr ) 
	\notag \\
	\dot{n} &= 2 \, \mathfrak{h} \wedge n 
	\notag 
\end{align}
on extended phase space $\Sigma = \R^{2d} \times \Stwo$. To shorten the notation, we will often use $z = (q,p) \in \R^{2d}$ for the translational variables and $\partial_z$ stands for either $\partial_{q_j}$ or $\partial_{p_j}$. Then the flow $\Phi^{\eps}_t = \bigl ( Z^{\eps}(t) , N^{\eps}(t) \bigr )$ similarly splits into a translational part $Z^{\eps}(t) = \bigl ( Q^{\eps}(t) , P^{\eps}(t) \bigr )$ and spin $N^{\eps}(t)$. It is easy to see that under the assumptions placed on $h$, the flow $\Phi^{\eps}$ exists globally in time and is smooth (Proposition~\ref{flow_estimates:prop:existence_flow}). 

If we replace $h$ by the leading-order term $h_0$ in the first two equations of \eqref{Egorov:eqn:hamiltonian_eom}, we obtain another flow $\Phi^0$ that  exists for all $t \in \R$, is smooth and agrees with $\Phi^{\eps}_t = \Phi^0_t + \order(\eps)$ to leading order for all bounded times. We will write $\Phi^0_t = \bigl ( Z^0(t) , N^0(t) \bigr )$ for the translational and spin part. Here, translational and spin dynamics decouple and there is no back-reaction from the spin dynamics onto the translational dynamics, the spin is just “dragged along”.  

Now we have the necessary terminology to prove our semiclassical limits: combining standard Weyl and Stratonovich-Weyl calculus with standard arguments, we obtain an Egorov theorem for times of order $1$. The proofs of the relevant properties of the flow are postponed to Section~\ref{flow_estimates}. 
\begin{theorem}[$\order(1)$-time Egorov-type theorem]\label{theorem:order_1_Egorov}
	Suppose $h$ satisfies $\mathfrak{h}_0 \in S^0$ in addition to Assumption~\ref{assumption:hamiltonian_symbol}. Then for any $a \in S^0_{\Sigma}$ and $T < \infty$, the following two statements hold: 
	\begin{enumerate}[(i)]
		\item $\displaystyle \sup_{t \in [-T,T]} \Bnorm{\e^{+ \ii \hat{h} \frac{t}{\eps}} \, \Opsum(a) \,\e^{-\ii \hat{h} \frac{t}{\eps}} - \Opsum \bigl ( a \circ \Phi^{\eps}_t \bigr )}_{\mathcal{B}(\Hil)} = \order(\eps)$
		\item If in addition $a$ is independent of $n$, then the error is of second order, 
		\begin{align}
			\sup_{t \in [-T,T]} \Bnorm{\e^{+ \ii \hat{h} \frac{t}{\eps}} \, \Opsum(a) \,\e^{-\ii \hat{h} \frac{t}{\eps}} - \Opsum \bigl ( a\circ \Phi^{\eps}_t \bigr )}_{\mathcal{B}(\Hil)} = \order(\eps^2) 
			. 
			\notag 
		\end{align}
	\end{enumerate}
\end{theorem}
\begin{proof}
	\begin{enumerate}[(i)]
		\item By Proposition~\ref{flow_estimates:prop:existence_flow}, the flow $\Phi^{\eps}$ exists and all its derivatives are bounded for all $\abs{t} \leq T$. We abbreviate the classically evolved observable with $a(t) := a \circ \Phi^{\eps}_t$ and set $\hat{a} := \Opsum(a)$. Then a Duhamel argument yields 
		\begin{align}
			\E^{\I \hat h \frac{t}{\epsi}}\,\hat{a} \,\E^{-\I \hat h \frac{t}{\epsi}}\;-\; \widehat{a(t)}
			&= \int_0^t \D s \, \tfrac{\D}{\D s} \Bigl ( \E^{\I \frac{s}{\epsi} \hat{h}} \, \widehat{a(t - s)} \, \E^{- \I \frac{s}{\epsi} \hat{h}} \Bigr )
			\notag 
			\\
			&= \int_0^t \D s \, \E^{ \I \frac{s}{\epsi} \hat{h}} \, \biggl ( \tfrac{\I}{\epsi} \bigl [ \hat{h} , \widehat{a(t-s)} \bigr ] 
			+ \tfrac{\D}{\D s} \,\widehat{a(t-s)} \biggr ) \, \E^{- \I \frac{s}{\epsi} \hat{h}} 
			\notag \\
			&= \int_0^t \D s \, \E^{ \I \frac{s}{\epsi} \hat{h}} \, \biggl ( \tfrac{\I}{\epsi} \bigl [ \hat{h} , \widehat{a(t-s)} \bigr ] 
			- \Opsum \bigl \{ h, a(t-s) \bigr \}_\Sigma \biggr ) \, \E^{- \I \frac{s}{\epsi} \hat{h}} 
			.
			\label{eqn:Duhamel_argument}
		\end{align}
		Combining Corollary~\ref{corollary:flow_preserves_Hoermander} with the usual Caldéron-Vaillancourt theorem \cite[Théorème~II~36]{Robert:tour_semiclassique:1987} and Lemma~\ref{lemma:commutator_to_Poisson}~(i), we obtain 
		\begin{align*}
			\sup_{t \in [-T,T]} \norm{\tfrac{\ii}{\eps} \bigl [ \hat{h} , \widehat{a(t-s)} \bigr ] - \Opsum \bigl \{ h, a(t-s) \bigr \}_{\Sigma}}_{\mathcal{B}(\Hil)} = \order(\eps)
		\end{align*}
		and thus we have shown (i). 
		\item If $a(q,p,n) = a(q,p)$, then $a_1(t) := a \circ \Phi^{\eps}_t - a \circ \Phi^0_t = \order(\eps)$ by Proposition~\ref{flow_estimates:prop:existence_flow} where $\Phi^0$ is the decoupled flow introduced in the beginning of the section. Thus Lemma~\ref{lemma:commutator_to_Poisson}~(ii) applies and 
		\begin{align*}
			\sup_{t \in [-T,T]} \norm{\tfrac{\ii}{\eps} \bigl [ \hat{h} , \widehat{a(t-s)} \bigr ] - \Opsum \bigl \{ h, a(t-s) \bigr \}_{\Sigma}}_{\mathcal{B}(\Hil)} = \order(\eps^2)
			.
		\end{align*}
		holds which in turn implies (ii). 
	\end{enumerate}
\end{proof}
Note that (i) just shows that we can replace the skew product flow \eqref{eqn:skew_product_flow} with the Hamiltonian flow $\Phi^{\eps}$ without changing the size of the error. For purely translational observables one can improve the error estimate by a factor of $\eps$ when going from the skew product flow to $\Phi^{\eps}$. 

As an immediate corollary, we obtain a semiclassical limit for states: here we compare $\hat{\rho}(t) = \e^{-\ii \frac{t}{\eps} \hat{h}} \, \hat{\rho} \, \e^{+\ii \frac{t}{\eps} \hat{h}}$ with $\Wignersum(\hat{\rho}) \circ \Phi^{\eps}_{-t}$, where $\Wignersum(\hat{\rho})$ is the Wigner transform of the density operator $\hat{\rho}$. 
\begin{corollary}
	In addition to the conditions of Theorem~\ref{theorem:order_1_Egorov}, assume that $\hat{\rho}$ is a density operator with Wigner transform $W := \Wignersum(\hat{\rho})$. 
	\begin{enumerate}[(i)]
		\item If $a \in S^0_{\Sigma}$ \emph{and} $W$ depend non-trivially on spin, then 
		\begin{align*}
			\mathrm{Tr} \bigl ( \Opsum(a) \, \hat{\rho}(t) \bigr ) &= \frac{1}{(2\pi)^d} \frac{1}{4 \pi} \int_{\R^{2d}} \dd q \, \dd p \int_{\Stwo} \dd n \, a(q,p,n) \, W \circ \Phi^{\eps}_t(q,p,n) + \order(\eps) 
		\end{align*}
		holds for all $\abs{t} \leq T$. 
		\item If $a \in S^0_{\Sigma}$ or $W$ is independent of $n$, then 
		\begin{align*}
			\mathrm{Tr} \bigl ( \Opsum(a) \, \hat{\rho}(t) \bigr ) &= \frac{1}{(2\pi)^d} \frac{1}{4 \pi} \int_{\R^{2d}} \dd q \, \dd p \int_{\Stwo} \dd n \, a(q,p,n) \, W \circ \Phi^{\eps}_t(q,p,n) + \order(\eps^2) 
		\end{align*}
		holds for all $\abs{t} \leq T$. 
	\end{enumerate}
\end{corollary}
\begin{proof}
	The proofs of (i) and (ii) rely on Theorem~\ref{theorem:order_1_Egorov}, 
	\begin{align*}
		\mathrm{Tr} \bigl ( \Opsum(a) \, \hat{\rho}(t) \bigr ) &= \mathrm{Tr} \left ( \Opsum(a) \, \e^{- \ii \frac{t}{\eps} \hat{h}} \, \hat{\rho} \, \e^{+ \ii \frac{t}{\eps} \hat{h}} \right ) 
		= \mathrm{Tr} \left ( \e^{+ \ii \frac{t}{\eps} \hat{h}} \, \Opsum(a) \, \e^{- \ii \frac{t}{\eps} \hat{h}} \, \hat{\rho} \right ) 
		, 
	\end{align*}
	the observation that $\hat{\rho}(t)$ and $W \circ \Phi^{\eps}_t$ solve the equations of motion for observables backwards in time and Liouville's theorem.
\end{proof}
Attempting to extend Egorov theorems to longer time scales is hard and requires more detailed information on the flow. The non-linearity of the classical system limits the time scale of validity of semiclassical approximations to the Ehrenfest time $t = \order(\sabs{\ln \eps})$.
While for quadratic Hamiltonians on phase-space $\R^{2d}$ one finds that $\hat{A}(t) = \Op \bigl ( A \circ \Phi^0_t \bigr) $ holds without error since the equations of motion are \emph{linear}, any coupling between translational and spin degrees of freedom introduces additional non-linear terms for the equations of motion on \emph{extended} phase space. However, in the case that $h_0$ is a quadratic polynomial in the components of $p$ and $q$, and the components of $\mathfrak{h}_0$ and $\mathfrak{h}$ are linear polynomials in $q$ and $p$, we can still find bounds on the derivatives of the flow for times of order $\order(\eps^{-\gamma})$ for some $\gamma > 0$ (\cf Proposition~\ref{proposition:long_time_flow_estimates}).  The Jaynes-Cummings- and Rabi-type Hamiltonians which, among other things, describe the interaction between an electromagnetic mode in a cavity and an atomic two-level system, are of this type.
 
A second important ingredient in our long-time semiclassical limit involves using suitable boundedness criteria for $\Psi$DOs: the usual Caldéron-Vaillancourt theorem \cite[Théorème~II~36]{Robert:tour_semiclassique:1987} requries us to control $2d+1$ derivatives where $d$ is the dimension of translational configuration space, \ie there exists a constant $c_d > 0$ depending only on $d$ such that 
\begin{align*}
	\bnorm{\Opsum(a)}_{\mathcal{B}(\Hil)} \leq c_d \, \snorm{a}_{0,2d+1} 
\end{align*}
holds for all $a \in S^0_{\Sigma}$ uniformly in $\eps \in (0,\eps_0)$. To reduce the number of derivatives, we will use other bounds: if $a$ decays sufficiently rapidly at $\infty$, we only need to control $d+1$ derivatives. More specifically, a straightforward generalization of \cite[Lemma~1.1]{Gerard_Leichtnam:ergodic_properties_eigenfunctions:1993} and \cite[Lemma~3.1]{Fermanian-Kammerer_Gerard_Lasser:Wigner_measure_propagation:2012} to matrix-valued symbols gives the following estimate: 
\begin{align}
	\bnorm{\Opsum(a)}_{\mathcal{B}(\Hil)} \leq c_d \, \max_{\abs{\alpha} \leq d+1} \sup_{(q,p,n) \in \Sigma} \Bigl ( \sexpval{p}^{d+1} \, \babs{\partial_p^{\alpha} a(q,p,n)} \Bigr )\,.
	\label{Egorov:eqn:boundedness_criterion_Gerard_Leichtnam}
\end{align}
Again, the constant $c_d > 0$ depends only on $d$ and is uniform in $\eps \in (0,\eps_0)$. Finally, for compactly supported symbols $a \in \Cont^{\infty}_{\mathrm{c}}(\Sigma)$, we can give a bound which only involves the sup norm, 
\begin{align}
	\bnorm{\Opsum(a)}_{\mathcal{B}(\Hil)} \leq \frac{c_d}{\eps^{\nicefrac{d}{2}}} \, \mathrm{Vol} \bigl ( \mathrm{supp} \, a \bigr ) \, \sup_{(q,p,n) \in \Sigma} \babs{a(q,p,n)} 
	\label{Egorov:eqn:boundedness_criterion_compact_support}
	. 
\end{align}
It is the last two boundedness criteria which enter the proof of the long-time semiclassical limit. 
\begin{theorem}[Long-time Egorov theorem]\label{theorem:long_time_Egorov}
	Let $d = 1$, $h_0(q,p) = \frac{1}{2} \bigl ( p^2 + \omega^2 q^2 \bigr )$, $\mathfrak{h}_0 = 0$ and $\mathfrak{h}$ be a linear polynomial in $q$ and $p$. Then for any $T < \infty$, the following statements holds: 
	\begin{enumerate}[(i)]
		\item Suppose $a \in S^{-2}_{\Sigma}$. Then for any $\gamma < \nicefrac{1}{4}$, there is $\eps_0 > 0$ such that for $\eps < \eps_0$ 
		\begin{align}
		\sup_{\abs{t}\in[0,T/\epsi^\gamma ]}	 \Bnorm{\e^{+\ii \frac{t}{\eps} \hat{h}} \, \Opsum(a) \, \e^{-\ii \frac{t}{\eps} \hat{h}} - \Opsum \bigl ( a \circ \Phi^{\eps}_t \bigr )}_{\mathcal{B}(\Hil)} = \order(\eps^{1 - 4 \gamma}) 
			. 
		\end{align}
		\item If in addition $a \in \Cont^{\infty}_{\mathrm{c}}(\Sigma)$ is independent of $n$, then for any $\gamma < \nicefrac{1}{2}$ there is $\eps_0 > 0$ such that for $\eps < \eps_0$ we have 
		\begin{align}
			\sup_{\abs{t}\in[0,T/\epsi^\gamma ]}	 \Bnorm{\e^{+\ii \frac{t}{\eps} \hat{h}} \, \Opsum(a) \, \e^{-\ii \frac{t}{\eps} \hat{h}} - \Opsum \bigl ( a \circ \Phi^{\eps}_t \bigr )}_{\mathcal{B}(\Hil)} = \order(\eps^{\nicefrac{3}{2} - 3 \gamma}) 
			. 
		\end{align}
	\end{enumerate}
\end{theorem}
\begin{proof}
	\begin{enumerate}[(i)]
		\item Assume $t = \order(\eps^{-\gamma})$ for some $\gamma \geq 0$ that has yet to be determined. Moreover, let $b , g > 0$ be the constants defined in the proof of Proposition~\ref{proposition:long_time_flow_estimates} and $\alpha = \nicefrac{1}{2}$. Lastly, for better readability, define the remainder 
		\begin{align*}
			R \bigl ( a(t) \bigr ):=- \Opspin \bigl \{ h_1, (1-\mathcal{P}) a(t) \bigr \}_{\R^{2d}} 
		\end{align*}
		which involves first-order derivatives of $a(t)$ in $q$ and $p$. Then the chain rule 
		%
		\begin{align*}
			 \partial_z \bigl ( a \circ \Phi^{\eps}_t \bigr )(z,n) = D a \bigl ( \Phi^{\eps}_t(z,n) \bigr ) \, \partial_z \Phi^{\eps}_t(z,n)
		\end{align*}
		and analogous expressions for higher-order derivatives, equation~\eqref{eqn:Duhamel_argument}, Lemma~\ref{lemma:commutator_to_Poisson}~(iii) and boundedness criterion~\eqref{Egorov:eqn:boundedness_criterion_Gerard_Leichtnam} imply 
		\begin{align*}
			\sup_{\abs{t}\in[0,T/\epsi^\gamma ]}	&\Bnorm{\e^{+\ii \frac{t}{\eps} \hat{h}} \, \Opsum(a) \, \e^{-\ii \frac{t}{\eps} \hat{h}} - \Opsum \bigl ( a \circ \Phi^{\eps}_t \bigr )}_{\mathcal{B}(\Hil)} 
			\leq 
			\\
			&\qquad \qquad \leq \eps \, \int_0^t \dd s \, \Bnorm{\Optrans \Bigl ( R \bigl ( a(t) \bigr ) \Bigr )}_{\mathcal{B}(\Hil)} 
			\\
			&\qquad \qquad \leq \eps \, \abs{t} \, C \, \max_{\abs{\alpha \leq 3}} \, \sup_{(z,n) \in \Sigma} \Bigl ( \sexpval{p}^2 \bnorm{\partial_z^{\alpha} \bigl ( a \circ \Phi^{\eps}_t \bigr )(z,n)} \Bigr ) 
			. 
		\end{align*}
		Proposition~\ref{proposition:long_time_flow_estimates} allows us to estimate the right-hand side by 
		\begin{align*}
			\max_{\abs{\alpha} \leq 3} \, \sup_{(z,n) \in \Sigma} \Bigl ( \sexpval{p}^2 \bnorm{\partial_z^{\alpha} \bigl ( a \circ \Phi^{\eps}_t \bigr )(z,n)} \Bigr ) \leq C \, \bigl ( 1 + \abs{t}^3 \bigr )
		\end{align*}
		for $\abs{t} \leq \eps^{- \nicefrac{1}{2}} \, \sqrt{2 b g}^{\, -1}$. Hence, for $t = \order(\eps^{- \gamma})$, $\gamma < \nicefrac{1}{4}$, we obtain the bound 
		\begin{align*}
			\sup_{\abs{t}\in[0,T/\epsi^\gamma ]}	 &\Bnorm{\e^{+\ii \frac{t}{\eps} \hat{h}} \, \Opsum(a) \, \e^{-\ii \frac{t}{\eps} \hat{h}} - \Opsum \bigl ( a \circ \Phi^{\eps}_t \bigr )}_{\mathcal{B}(\Hil)} 
			\leq \eps \, C \, \bigl ( \abs{t} + \abs{t}^4 \bigr ) 
			= \order(\eps^{1 - 4 \gamma})
			. 
		\end{align*}
		\item Again let $t = \order(\eps^{-\gamma})$ for some $\gamma \geq 0$ that has yet to be determined. In case $a$ is initially independent of $n$ and has compact support, we have $a(t,z,n) = a \bigl ( Z^{\eps}(t,z,n) \bigr )$ and thus 
		\begin{align*}
			\partial_z a(t,z,n) = D a \bigl ( Z^{\eps}(t,z,n) \bigr ) \, \partial_z Z^{\eps}(t,z,n) 
			. 
		\end{align*}
		Since the volume measure on phase space coincides with the Liouville measure, the above equation combined with Liouville's Theorem 
		implies that for all $t \in \R$, we have 
		\begin{align*}
			\mathrm{Vol} \bigl ( \mathrm{supp} \, \partial_z a(t) \bigr ) \leq \mathrm{Vol} \bigl ( \mathrm{supp} \, a(t) \bigr ) 
			= \mathrm{Vol} \bigl ( \mathrm{supp} \, a \bigr ) 
			. 
		\end{align*}
		The function $a_0(t) := a \circ \Phi^0_t$ is independent of $n$ and thus $\projs a_0(t) = a_0(t)$ holds for all $t$. This means, we can insert $a_0(t)$ into the remainder and estimate its norm for times $\abs{t} \leq \eps^{-\nicefrac{1}{2}} \, \sqrt{2 b g}^{\, -1}$ using equation~\eqref{Egorov:eqn:boundedness_criterion_compact_support} and Proposition~\ref{proposition:long_time_flow_estimates}, 
		\begin{align*}
			\eps \, &\abs{t} \, \sup_{\abs{s} \in [0,t]} \Bnorm{\Optrans \Bigl ( R \bigl ( a(t) - a_0(t) \bigr ) \Bigr )}_{\mathcal{B}(\Hil)} 
			\leq 
			\\
			&\leq \eps \, \abs{t} \, \frac{C}{\sqrt{\eps}} \, \max_{\abs{\alpha} \leq 1}\sup_{\abs{s} \in [0,t]} \sup_{(z,n) \in \Sigma} \Bigl ( 
			\babs{Da \bigl ( Z^{\eps}(s,z,n) \bigr ) - D a \bigl ( Z^0(s,z) \bigr )} \, \babs{\partial_z^{\alpha} Z(s,z,n)} 
			+ \Bigr . \\
			&\qquad \qquad \qquad \qquad \qquad \qquad \quad \;
			+ \babs{D a \bigl ( Z^0(s,z) \bigr )} \, \babs{\partial_z^{\alpha} \bigl ( Z^{\eps}(s,z,n) - Z^0(s,z) \bigr )} 
			\Bigr )
			\\
			&\leq \sqrt{\eps} \, C \, \abs{t} \, \bigl ( \eps \abs{t} + \eps \, b g \abs{t}^3 \bigr ) 
			= \order(\eps^{\nicefrac{3}{2} - 3 \gamma})
			. 
		\end{align*}
		Hence, as long as $t = \order(\eps^{- \gamma})$ for $\gamma < \nicefrac{1}{2}$, the right-hand side goes to $0$ as $\eps \to 0$ and we have shown (ii). 
	\end{enumerate}
\end{proof}
In the special case of a quadratic Hamiltonian we can actually get an exact equation for the time-evolution of the symbol of a semiclassical operator, albeit not in terms of a classical flow.
\begin{theorem}\label{exact}
	Let $d=1$ and $h_0(q,p) = \frac{1}{2}(p^2 + \omega^2 q^2)$, $\mathfrak{h}_0=0$ and $\mathfrak{h}$ be a polynomial in $q$ and $p$ of degree $1$. Let $a\in S^{-2}_\Sigma$ and $a(t,q,p,n)$ be a solution of 
	\begin{align*}
		\partial_t a(t,q,p,n) = \bigl ( \mathcal{P} \, \{h, a(t) \}_\Sigma \bigr ) \,(q,p,n) 
	\end{align*}
	such that $a(t) \in S^{0}_\Sigma$ for all $t \in \R$. Then
	\begin{align}\label{egorov3}
		 \E^{\I \hat h \frac{t}{\epsi}} \, \hat{a} \, \E^{-\I \hat h \frac{t}{\epsi}} =  \Opsum \bigl ( a(t) \bigr )
		.
	\end{align}
\end{theorem}
\begin{proof}
	By definition $a(t,q,p,\cdot)\in C_1(\Stwo)$ for all $t\in\R$. Thus by Lemma~\ref{lemma:commutator_to_Poisson} (iii) we have
	\begin{align*}
		\tfrac{\I}{\epsi} \bigl [ \Opsum h, \Opsum a(t) \bigr ] = \Opsum \{h,a(t)\}_\Sigma =\Opsum \mathcal{P} \{h,a(t)\}_\Sigma 
		= \Opsum \bigl ( \dot a(t) \bigr )
		.
	\end{align*}
\end{proof}

\section{Estimates on the classical flow}\label{flow}
\label{flow_estimates}
In this section, we study properties of the flows $\Phi^{\eps}$ and $\Phi^0$ as defined in the beginning of Section~\ref{Egorov}: $\Phi^{\eps}_t = \bigl ( Z^{\eps}(t) , N^{\eps}(t) \bigr )$ is the Hamiltonian flow associated to \eqref{Egorov:eqn:hamiltonian_eom} while $\Phi^0_t = \bigl ( Z^0(t) , N^0(t) \bigr )$ is the flow associated to \eqref{Egorov:eqn:hamiltonian_eom} after replacing $h$ by $h_0$ in the first two equations. Existence for all times and smoothness follows from standard arguments from the theory of ordinary differential equations. 
\begin{proposition}\label{flow_estimates:prop:existence_flow}
	Suppose $h$ satisfies $\mathfrak{h} \in S^0$ in addition to Assumption~\ref{assumption:hamiltonian_symbol}. Then $\Phi^{\eps}$ and $\Phi^0$ exist globally in time and are smooth, \ie for any $t \in \R$, $\Phi^{\eps}$ and $\Phi^0$ are diffeomorphisms on $\Sigma$ and depend smoothly on time. Moreover, for bounded time intervals $[-T,+T]$, we have $\Phi^{\eps}_t = \Phi^0_t + \order(\eps)$ and all partial derivatives $\partial_z^{\alpha} \partial_n^{\beta} \Phi^{\eps}_t$ and $\partial_z^{\alpha} \partial_n^{\beta} \Phi^0_t$ are uniformly bounded. 
\end{proposition} 
\begin{proof}
	All claims follow from the fact that the vector fields (r.h.s.~of \eqref{Egorov:eqn:hamiltonian_eom}) have globally bounded derivatives to all orders. Moreover, the difference between the vector fields which define $\Phi^{\eps}$ and $\Phi^0$ is bounded and $\order(\eps)$, and thus the Grönwall lemma implies $\Phi^{\eps}_t = \Phi^0_t + \order(\eps)$ for bounded time intervals. 
\end{proof}
This immediately implies that the actions of $\Phi^{\eps}$ and $\Phi^0$ preserve Hörmander classes. 
\begin{corollary}\label{corollary:flow_preserves_Hoermander}
	Let $h$ satisfy Assumption~\ref{assumption:hamiltonian_symbol} and, in addition, let $\mathfrak{h} \in S^0$. Then for $a \in S_\Sigma^k(\epsi)$ it holds that $a \circ \Phi^{\eps}_t , a \circ \Phi^0_t \in S_\Sigma^k(\eps)$ for all $t \in \R$ and for any $T < \infty$ and $r\in\N_0$
	\begin{align*}
		\sup_{t \in [-T,T]} \sup_{\eps \in (0,\eps_0)} \bnorm{a \circ \Phi^{\eps}_t}_{k,r} < \infty 
		\qquad \mbox{and}\qquad
		\sup_{t \in [-T,T]} \sup_{\eps \in (0,\eps_0)} \bnorm{a \circ \Phi^0_t}_{k,r} < \infty 
		. 
	\end{align*}
\end{corollary} 
Now let us turn to the case of Jaynes-Cummings- and Rabi-type Hamiltonians. If $\mathfrak{h}$ is the prefactor of $h_1 = \sqrt{3} \, \mathfrak{h} \cdot n$, we define the skew symmetric matrix
\begin{align*}
	\mathfrak{H} := 2 \, \left (
	\begin{matrix}
		0 & - \mathfrak{h}_3 & + \mathfrak{h}_2 \\
		+ \mathfrak{h}_3 & 0 & - \mathfrak{h}_1 \\
		- \mathfrak{h}_2 & + \mathfrak{h}_1 & 0 \\
	\end{matrix}
	\right ) 
	. 
\end{align*}
This convention allows us to write down proofs of the long-time flow estimates in a more compact fashion. 
\begin{proposition}\label{proposition:long_time_flow_estimates}
	Let $d = 1$ and 
	\begin{align*}
		h_0(q,p) = \tfrac{1}{2} \bigl (p^2 + \omega^2 q^2 \bigr )
		\, , \qquad 
		h_1(q,p,n) = \sqrt{3} \, \bigl ( \mathfrak{h}_c +q \, \mathfrak{h}_q +p \, \mathfrak{h}_p \bigr ) \cdot n 
		,
	\end{align*}
	where $\mathfrak{h}_c , \mathfrak{h}_q , \mathfrak{h}_p \in \R^3$. Then the flows $\Phi^{\eps}$ and $\Phi^0$ exist globally in time and the two are $\order(\eps)$-close, 
	\begin{align*}
		\sup_{(z,n) \in \Sigma} \babs{Z^{\eps}(t,z,n) - Z^0(t,z))} &= \order(\eps \abs{t})
		, \\ 
		\sup_{(z,n) \in \Sigma} \babs{N^{\eps}(t,z,n)- N^0(t,z,n)} &= \order(\eps \abs{t}^2) 
		. 
	\end{align*}
	We abbreviate $Z = Z^{\eps}$ and $N = N^{\eps}$ and introduce the norms 
	\begin{align*}
		\snorm{X}_t := \sup_{s \in [-t,t]} \sabs{X(s)} 
		.
	\end{align*}
	Let $'$ denote a derivative with respect to either $q$ or $p$. Then there are constants $b , g > 0$ depending only on $h_1$ such that for any $0 < \alpha < 1$ and all $\abs{t} \leq \eps^{- \nicefrac{1}{2}} \, \sqrt{\frac{\alpha}{bg}}$ it holds that
	\begin{align*}
		\bnorm{Z'}_t &\leq \frac{1}{1-\alpha} 
		, &&
		\bnorm{N'}_t \leq \frac{b |t|}{1-\alpha} 
		, 
		\\
		\bnorm{Z''}_t &\leq \frac{\alpha b |t|}{(1-\alpha)^3}
		, &&
		\bnorm{N''}_t \leq \frac{b^2 |t|^2}{(1-\alpha)^3} 
		, 
		\\
		\bnorm{Z'''}_t &\leq \frac{4 \alpha b^2 |t|^2}{(1-\alpha)^5}
		, &&
		\bnorm{N'''}_t \leq \frac{4 b^3 |t|^3}{(1-\alpha)^5} 
	\end{align*}
	and
	\begin{align*}
		\bnorm{Z' -Z^{0'}}_t \leq \frac{\eps \, b g |t|^2}{1-\alpha} 
		. 
	\end{align*}
\end{proposition} 
\begin{proof}
 To write the equations of motion in a concise form, we abbreviate
\begin{align*}
	\Omega = \left (
	\begin{matrix}
		0 & 1 \\
		- \omega^2 & 0 \\
	\end{matrix}
	\right )
	&&
	\mbox{and} 
	&&
	G = \sqrt{3} \, \left (
	\begin{matrix}
		+ \mathfrak{h}_p^T \\
		- \mathfrak{h}_q^T \\
	\end{matrix}
	\right ) 
	. 
\end{align*}
Then
\begin{align}
    \dot{Z} = \Omega Z +\epsi G N 
	\, ,
	\qquad
    \dot{N} = \mathfrak{H}(Z) N = \mathfrak{H}(Z^0) N + \mathfrak{H}(Z-Z^0) N 
	\, .
	\label{harmHEM}
\end{align}
Let $R_Z(t)\in SO(2)$ be the propagator of the harmonic oscillator satisfying $\dot{R}_Z(t) = \Omega R_Z(t)$ with $R_Z(0)= \id_{\R^2}$, then the ``variation of constants'' formula gives
\begin{align*}
	Z (t) = R_Z(t) \left( Z (0) + \eps \, \int_0^t \dd s \, R_Z(-s) \, G \, N (s) \right ) 
	= Z^0 (t) + \eps \int_0^t \dd s \, R_Z(t-s) \, G \, N (s) 
	.
\end{align*}
Thus $\sup_{(q,p,n)\in\Sigma } \babs{Z (t,q,p,n)-Z^0(t,q,p)} = \Or( \epsi |t|)$. 
For $N$ we find analogously with $A_0(t)\in SO(3)$ the propagator of the ``homogenous'' system
\begin{align*}
	\dot{A}_0(t) = \mathfrak{H} \bigl ( Z^0(t) \bigr ) A_0(t)
\end{align*}
that 
\begin{align*}
	N(t) &= A_0(t) \left ( N (0) + \int_0^t \dd s \, A_0(-s) \, \mathfrak{H} \bigl ( Z(s)-Z^0(s) \bigr ) \, N (s) \right ) 
	\\
	&= N^0 (t) + \int_0^t \dd s \, A_0(t-s) \, \mathfrak{H} \bigl ( Z(s)-Z^0(s) \bigr ) \, N (s) 
\end{align*}
and thus $\sup_{(q,p,n)\in\Sigma } \babs{N (t,q,p,n)-N^0(t,q,p,n)} = \order(\eps |t|^2)$.

We now turn to the derivatives of the flow. Let $'$ be a derivative with respect to the initial $q$ or $p$. Then the derivatives of the flow map satisfy the equations
\begin{align*}
	\dot{Z}' = \Omega Z' + \eps G N'
	\qquad
	\dot{N}' = \mathfrak{H}(Z) N' + B(Z') N\,,
\end{align*}
where $B(Z')$ is a skew-symmetric matrix with all components linear in $Q'$ and $P'$. Let $R_Z(t)\in SO(2)$ be as before and $R_N(t)\in SO(3)$ be the propagator of the ``homogenous'' system
\begin{align*}
	\dot{R}_N(t) = \mathfrak{H} \bigl ( Z(t) \bigr ) R_N(t) 
	,
\end{align*}
then the variation of constants formula and $\Phi^{\eps}_0(z,n) = (z,n)$ gives
\begin{align*}
	Z'(t) &= R_Z(t) \left ( Z'(0) + \epsi \int_0^t \dd s \, R_Z(-s) \, G \, N'(s) \right )
	\\
	N'(t) &= R_N(t) \biggl ( \underbrace{N'(0)}_{=0} + \int_0^t \dd s \, R_N(-s) \, B \bigl ( Z'(s) \bigr ) \, N(s) \biggr ) 
	\\
	&
	= \int_0^t \dd s \, R_N(t-s) \, B \bigl ( Z'(s) \bigr ) \, N(s) 
	.
\end{align*}
From this we infer that 
\begin{align*}
   \bnorm{Z'}_t &\leq \underbrace{\babs{Z'(0)}}_{=1} + \eps \, g |t| \, \bnorm{N'}_t 
	= 1 + \eps \, g |t| \, \bnorm{N'}_t 
    \\
    \bnorm{N'}_t &\leq b |t| \, \bnorm{Z'}_t 
\end{align*}
where $g=\|G\|_{\mathcal{B}(\R^3,\R^2)}$ and $b$ is a constant depending on $B$. Combining these estimates we find that for any $\alpha<1$
\begin{align*}
	\bnorm{Z'}_t \leq \frac{1}{1 - \eps \, b g |t|^2} 
	\leq \frac{1}{1-\alpha} 
	\quad\mbox{and}\quad 
	\bnorm{N'}_t \leq \frac{b |t|}{1 - \eps \, b g |t|^2} 
	\leq \frac{b |t|}{1-\alpha} \,
\end{align*}
uniformly for $|t| \leq \sqrt{\frac{\alpha}{\eps gb}}$. 
Similarly, from 
\begin{align*}
	\dot{Z}' - \dot{Z}^{0'} = \eps G N' 
	, 
\end{align*}
and the norm estimate for $N'$, we deduce 
\begin{align*}
	\bnorm{Z' - Z^{0'}}_t \leq \eps \, g \abs{t} \, \bnorm{N'}_t 
	\leq \frac{\eps \, b g \, \abs{t}^2}{1 - \alpha}
	. 
\end{align*}
Bounds on higher derivatives are obtained in the same way: 
\begin{align*}
	\dot{Z}'' &= \Omega Z'' +\epsi G N''
	\\
	\dot{N}'' &= \mathfrak{H}(Z) N'' + B(Z'') N + B(Z') N' 
\end{align*}
implies
\begin{align*}
	\bnorm{Z''}_t &\leq \babs{Z''(0)} + \eps \, g |t| \, \bnorm{N''}_t 
	= \eps \, g |t| \, \bnorm{N''}_t
	\\
	\bnorm{N''}_t &\leq b |t| \, \bigl ( \bnorm{Z''}_t + \bnorm{Z'}_t \, \bnorm{N'}_t \bigr ) 
	\leq b |t| \, \left ( \bnorm{Z''}_t + \frac{b \, |t|}{(1-\alpha)^2} \right )
\end{align*}
and thus 
\begin{align*}
	\bnorm{N''}_t \leq \eps \, b g |t|^2 \, \bnorm{N''}_t + \frac{b^2 \, |t|^2}{(1-\alpha)^2} \,.
\end{align*}
Solving for $\bnorm{N''}_t$ gives
\begin{align*}
	\bnorm{N''}_t \leq \frac{1}{1 - \eps \, b g |t|^2} \; \frac{b^2 |t|^2}{(1-\alpha)^2} 
	\leq \frac{b^2 |t|^2}{(1-\alpha)^3}
\end{align*}
and 
\begin{align*}
	\bnorm{Z''}_t \leq \frac{\eps \, b^2 g |t|^3}{(1-\alpha)^3} \leq \frac{\alpha b |t|}{ (1-\alpha)^3} 
	.
\end{align*}
For the third derivatives we find 
\begin{align*}
	\dot{Z}''' &= \Omega Z''' + \eps G N'''
	\\
	\dot{N}''' &= \mathfrak{H}(Z) N''' + B(Z''') N + 2 B(Z') N'' + 2 B(Z'') N' 
\end{align*}
and thus 
\begin{align*}
	\bnorm{Z'''}_t &\leq \eps \, g |t| \, \bnorm{N'''}_t
	\\
	\bnorm{N'''}_t &\leq b |t| \, \bigl ( \bnorm{Z'''}_t + 2 \, \bnorm{Z'}_t \, \bnorm{N''}_t + 2 \, \bnorm{Z''}_t \, \bnorm{N'}_t \bigr ) 
	\\
	&\leq b |t| \, \left ( \bnorm{Z'''}_t + \frac{2 b^2 |t|^2}{(1-\alpha)^4} + \frac{2 \eps \, b^3 g |t|^4}{(1-\alpha)^4} \right)
	\\
	&\leq 
	\eps \, b g |t|^2 \, \bnorm{N'''}_t + \frac{2 b^3 |t|^3 \, (1 + \eps \, b g |t|^2)}{ (1-\alpha)^4} 
	\\
	&\leq \eps \, b g |t|^2 \, \bnorm{N'''}_t + \frac{4 b^3 |t|^3}{(1-\alpha)^4} 
\end{align*}
where in the last step we have used $0 < \alpha < 1$ and $\abs{t} \leq \eps^{- \nicefrac{1}{2}} \, \sqrt{\frac{\alpha}{b g}}$. Solving again for $N'''$ we find
\begin{align*}
	\bnorm{N'''}_t \leq \frac{4 b^3 |t|^3}{(1-\alpha)^5} 
	\quad \mbox{and} \quad 
	\bnorm{Z'''}_t \leq \frac{4 \eps \, b^3 g |t|^4}{(1-\alpha)^5} 
	\leq \frac{4 \alpha b^2 |t|^2}{(1-\alpha)^5}
	.
\end{align*}
\end{proof}

\section{The Stern-Gerlach experiment} 
\label{Stern-Gerlach}
As an illustration of the method we discuss the Stern-Gerlach experiment:
Neutral atoms with magnetic moment $g$ and spin-$\nicefrac{1}{2}$ are sent through a weak, inhomogeneous magnetic field $\mathbf{B} = \bigl ( \mathbf{B}_1 , \mathbf{B}_2 , \mathbf{B}_3 \bigr ) \in \Cont^{\infty}_{\mathrm{b}}(\R^3,\R^3)$. For simplicity, we will absorb $g$ into $\mathbf{B}$. 
In the experiment one observes that a beam of such particles splits into two parts with intensities depending on the initial spin-state.
The Hamiltonian describing a single atom in the beam is the Weyl quantization of 
\begin{align}
	H(q,p) &= \tfrac{1}{2} p^2 - \eps \, \tfrac{1}{2} \mathbf{B}(q) \cdot \sigma 
	\\
	&
	= \Opspin \Bigl ( \tfrac{1}{2} p^2 - \eps \, \tfrac{\sqrt{3}}{2} \mathbf{B}(q) \cdot n \Bigr ) 
	=: \Opspin \bigl ( h(q,p,\cdot) \bigr ) 
	. 
	\notag 
\end{align}
Due to the assumption on the magnetic field, Theorem~\ref{theorem:order_1_Egorov} applies and for observables $a \in S^0_{\Sigma}$ which are initially independent of spin, we have a semiclassical limit with error $\order(\eps^2)$. In particular, this implies we can compute quantum expectation values 
\begin{align*}
	\mathrm{Tr}_{L^2(\R^3,\C^2)} &\Bigl ( \e^{+\ii \frac{t}{\eps} \hat{h}} \, \Opsum(a) \, \e^{- \ii \frac{t}{\eps} \hat{h}} \, \hat{w} \Bigr ) 
	= \\
	&
	= \frac{1}{(2\pi)^3} \int_{\R^6} \dd q \, \dd p \, \frac{1}{4\pi} \int_{\Stwo} \dd n \, a \circ \Phi^{\eps}_t(q,p,n) \, w(q,p,n) + \order(\eps^2)
\end{align*}
with respect to the state $\hat{w} = \Opsum(w)$ for times of order $1$. 

To be able to solve the semiclassical equations of motion \eqref{Egorov:eqn:hamiltonian_eom}  {analytically}, we will make some simplifying assumptions: first of all, we take the magnetic field to be of the form 
\begin{align*}
	\mathbf{B}(q) = \bigl ( 0 , 0 , b(q_1) \bigr ) 
	= b(q_1) \, e_3 
\end{align*}
where $b \in \Cont^{\infty}_{\mathrm{c}}(\R)$, and thus $\nabla_q \cdot \mathbf{B} = 0$, $\nabla_q \mathbf{B}_j = 0$ for $j = 1,2$ and $\nabla_q \mathbf{B}_3 = b' \, e_1$. It is easy to solve the equations of motion 
\begin{align}
	\dot{q} &= p 
	\notag \\
	\dot{p} &=   \eps \, \tfrac{\sqrt{3}}{2} \, \nabla_q \bigl ( \mathbf{B} \cdot n \bigr ) 
	=   \eps \, \tfrac{\sqrt{3}}{2} \, b' \, n_3 \, e_1 
	= \order(\eps) 
	\label{Stern-Gerlach:eqn:eom}
	\\
	\dot{n} &= b \, e_3 \wedge n 
	\notag 
\end{align}
explicitly up to $\order(\eps^2)$: the leading-order flow is given by 
\begin{align*}
	\Phi^0_t(q,p,n) = \left (
	\begin{matrix}
		q + t \, p \\
		p \\
		N^0(t,q,p,n) \\
	\end{matrix}
	\right )
\end{align*}
where $N^0(t,q,p,n)$ solves 
\begin{align*}
	\dot{N}^0(t,q,p,n) &= b(q_1 + t \, p_1) \, e_3 \wedge N^0(t,q,p,n) 
	, 
	\qquad \qquad 
	N^0(0,q,p,n) = n 
	. 
\end{align*}
Note that the spin precesses around the $e_3$-axis and thus $N_3^0(t,q,p,n) = n_3$. The flow which solves \eqref{Stern-Gerlach:eqn:eom} can be found up to $\order(\eps^2)$ by iteration: $\ddot{q} = \epsi  \frac{\sqrt{3}}{2} b' (q_1)\, n_3 \, e_1$ and hence, for initial momenta $p = (0,p_2,p_3)$ in the $p_2 p_3$-plane, we compute 
\begin{align*}
	q^{\eps}(t,q,p,n) &= q + t \, p +\eps \, t^2 \, \tfrac{\sqrt{3}}{2} \, b'(q_1) \, n_3 \, e_1 
	. 
\end{align*}
Thus, classical trajectories starting at $(q,p)$ split into a whole fan of possible directions depending on  $n_3 \in [-1,+1]$. This apparent inconsistency with the quantum mechanical predictions disappears after averaging over an initial distribution of spin. Let us assume for simplicity that the initial state is a product state, $\hat{w} = \hat{w}_{\R^6} \otimes \hat{w}_{\Stwo}$, with symbol $w(q,p,n) = w_{\R^6}(q,p) \, w_{\Stwo}(n)$. The fact that $w_{\Stwo}(n) = \frac{1}{2}\left(\mathfrak{s}_0 + \sqrt{3} \, \mathfrak{s} \cdot n\right)$ is a density matrix, \ie 
\begin{align*}
	\Opspin(w_{\Stwo}) &= \Opspin(w_{\Stwo})^* 
	, 
	&&
	0 \leq \Opspin(w_{\Stwo}) \leq 1 
	, 
	&&
	\mathrm{Tr}_{\C^2} \bigl ( \Opspin(w_{\Stwo}) \bigr ) = 1 
	, 
\end{align*}
implies $\mathfrak{s}_0 = 1$, $\mathfrak{s} \in \R^3$ and $\abs{\mathfrak{s}} \leq 1$. The pure states have ``spin-direction'' $\mathfrak{s}\in \Stwo$, while the completely unpolarized state is $w_{\Stwo}^0(n) = \frac{1}{2}$.
Taking the spin-average of the $e_1$-deflection  
\[
\Delta q^{\eps}_1(t) := q^\eps_1(t) - q_1 = \eps \, t^2 \, \tfrac{\sqrt{3}}{2} \, b'(q_1) \, n_3  
\]
for initial momenta $p = (0,p_2,p_3)$ with respect to $\Opspin(w_{\Stwo}) = \tfrac{1}{2} \left( \id_{\C^2} + \mathfrak{s} \cdot \sigma\right)$ yields 
\begin{align}
	\mathbb{E}_{w_{\Stwo}} \bigl ( \Delta q^{\eps}_1(t) \bigr ) :& \negmedspace= 
	\mathrm{Tr}_{\C^2} \Bigl ( \Opspin \bigl ( \Delta q^{\eps}_1(t) \bigr ) \, \Opspin(w_{\Stwo}) \Bigr ) 
	\notag \\
	&=    \eps \, t^2 \,  \frac{b'(q_1)}{2}    \, \mathrm{Tr}_{\C^2} \bigl ( \sigma_3 \, \Opspin(w_{\Stwo}) \bigr )  
	\notag \\
	&=     \eps \, t^2 \, \frac{ b'(q_1) }{2}  \, \mathfrak{s}_3\,   
	. 
	\label{Stern-Gerlach:eqn:averaged_position}
\end{align}
While this is in line with the quantum mechanical predictions, in order to rederive the fact that the distribution of $ \Delta q^{\eps}_1(t)$ is concentrated on two points, we need to compute also higher moments of $ \Delta q^{\eps}_1(t)$. To this end note that
\[
\left(  \Opspin \bigl ( \Delta q^{\eps}_1(t) \bigr )\right)^{2m} =  \left( \eps \, t^2 \,  \, \frac{|b'(q_1)|}{2}\right)^{2m}
\sigma_3^{2m} =  \left( \eps \, t^2 \,  \, \frac{|b'(q_1)|}{2}\right)^{2m} \id_{\C^2}
\]
and 
\[
\left(  \Opspin \bigl ( \Delta q^{\eps}_1(t) \bigr )\right)^{2m+1} =  \left( \eps \, t^2 \,  \, \frac{|b'(q_1)|}{2}\right)^{2m+1}
\sigma_3^{2m+1} =  \left( \eps \, t^2 \,  \, \frac{|b'(q_1)|}{2}\right)^{2m+1} \sigma_3\,.
\]
Hence
\[
	\mathbb{E}_{w_{\Stwo}} \bigl ( (\Delta q^{\eps}_1(t))^{2m} \bigr )  = 
	   \left( \eps \, t^2 \,  \, \frac{|b'(q_1)|}{2}\right)^{2m}
\]
and
\[
  \mathbb{E}_{w_{\Stwo}} \bigl ( (\Delta q^{\eps}_1(t))^{2m+1} \bigr )  = 
	   \left( \eps \, t^2 \,  \, \frac{|b'(q_1)|}{2}\right)^{2m+1}
	    \, \mathfrak{s}_3\,.
\]
In conclusion, we have that the distribution for the deflection is concentrated on the points $\pm  \eps \, t^2 \, \frac{ b'(q_1) }{2} $ with weights $\frac{1}{2}(1 \pm \mathfrak{s}_3)$, which is exactly the quantum mechanical prediction.

Of course we still need to average over an initial distribution $w_{\R^6}(q,p)$ of position and momenta. According to the uncertainty relation we can assume that the width of the initial position distribution is $\Delta q_0 =\Or( \epsi^\frac{1}{3})$ and  the width of the initial momentum distribution is $\Delta p_0 =\Or(\epsi^\frac{2}{3})$. After time $t \propto \epsi^{-\frac12}$ the  width $\Delta q = \Delta q_0 + t \Delta p_0  =\Or(  \epsi^\frac{1}{6})$ of wave packets following the classical trajectories is smaller than the separation
$\eps \, t^2 \,   b'(q_1)  =\Or(1)$
 of the trajctories due to the deflection by the inhomogeneous field. 

In conclusion, we see that the semiclassical model can correctly reproduce the splitting of wave packets in inhomogeneous fields when starting in   quantum mechanical states. For $b(q_1) = q_1$ we even have  a rigorous proof that the semiclassical predictions agree at leading order with the quantum mechnical distributions up to times where the splitting is visible.

\bibliographystyle{alpha}
\bibliography{bibliography}

\newcommand{\etalchar}[1]{$^{#1}$}
\begin{thebibliography}{CSMH{\etalchar{+}}10}

\bibitem[BG04]{Bolte_Glaser:semiclassical_limit_matrix-valued_operators:2004}
Jens Bolte and Rainer Glaser.
\newblock {A semiclassical Egorov theorem and quantum ergodicity for matrix
  valued operators}.
\newblock {\em Communications in Mathematical Physics}, 247:391--419, 2004.

\bibitem[BGK01]{Bolte_Glaser_Keppeler:ergodicity_spinning_particle:2001}
Jens Bolte, Rainer Glaser, and Stefan Keppeler.
\newblock {Quantum and classical ergodicity of spinning particles}.
\newblock {\em Annals of Physics}, 293:1--14, 2001.

\bibitem[BGP99]{Bambusi_Graffi_Paul:semiclassics_Ehrenfest:1999}
Dario Bambusi, Sandro Graffi, and Thierry Paul.
\newblock {Long time semiclassical approximation of quantum flows: A proof of
  the Ehrenfest time}.
\newblock {\em Asymptotic Analysis}, 21(2):149--160, 1999.

\bibitem[BK99]{Bolte_Keppeler:semiclassics_Dirac:1999}
Jens Bolte and Stefan Keppeler.
\newblock {A semiclassical approach to the Dirac equation}.
\newblock {\em Annals of Physics}, 274:125--162, 1999.

\bibitem[BR02]{Bouzouina_Robert:uniform_semiclassical_estimates_Ehrenfest:2002}
Abdelkader Bouzouina and Didier Robert.
\newblock {Uniform semiclassical estimates for the propagation of quantum
  observables.}
\newblock {\em Duke Math Journal}, 111(2):223--252, 2002.

\bibitem[CSMH{\etalchar{+}}10]{Chuchem_Smith-Mannschott_Hiller_Kottos_Vardi_Co%
hen:quantum_dynamics_josephson_junction:2010}
Maya Chuchem, Katrina Smith-Mannschott, Moritz Hiller, Tsampikos Kottos,
  Amichay Vardi, and Doron Cohen.
\newblock {Quantum dynamics in the bosonic Josephson junction}.
\newblock {\em Phys. Rev. A}, 82:053617, 2010.

\bibitem[EW96]{Emmrich:1996vv}
C.\ Emmrich and A.\ Weinstein.
\newblock {Geometry of the transport equation in multicomponent WKB
  approximations}.
\newblock {\em Communications in Mathematical Physics}, 176:701--711, 1996.

\bibitem[FKGL13]{Fermanian-Kammerer_Gerard_Lasser:Wigner_measure_propagation:2%
012}
Clotilde Fermanian-Kammerer, Patrick G\'erard, and Caroline Lasser.
\newblock Wigner measure propagation and conical singularity for general
  initial data.
\newblock {\em Archive for Rational Mechanics and Analysis}, 209(1):209--236,
  2013.

\bibitem[Fol89]{Folland:harmonic_analysis_hase_space:1989}
Gerald~B. Folland.
\newblock {\em {Harmonic Analysis on Phase Space}}.
\newblock Princeton University Press, 1989.

\bibitem[GL93]{Gerard_Leichtnam:ergodic_properties_eigenfunctions:1993}
Patrick Gérard and Eric Leichtnam.
\newblock {Ergodic properties of eigenfunctions for the Dirichlet problem}.
\newblock {\em Duke Math. Journal}, 71(2):559--607, 1993.

\bibitem[GLT13]{Gat:2013vm}
Omri Gat, Max Lein, and Stefan Teufel.
\newblock {Resonance phenomena in the interaction of a many-photon wave packet
  and a qubit}.
\newblock {\em J.\ Phys.\ A: Math.\ Theor.}, 46:315301, 2013.

\bibitem[Kg81]{Kumanogo:pseudodiff:1981}
Hitoshi Kumano-go.
\newblock {\em {Pseudodifferential Operators}}.
\newblock The MIT Press, 1981.

\bibitem[Lei10]{Lein:quantization_semiclassics:2010}
Max Lein.
\newblock {Weyl Quantization and Semiclassics}, 2010.
\newblock Technische Universität München.

\bibitem[LF92]{LITTLEJOHN:1992vg}
Robert~G.\ Littlejohn and William~G.\ Flynn.
\newblock {Semiclassical Theory of Spin-Orbit-Coupling}.
\newblock {\em Physical Review A}, 45(11):7697--7717, 1992.

\bibitem[Mar02]{Martinez:intro_microlocal_analysis:2002}
Andre Martinez.
\newblock {\em {An Introduction to Semiclassical and Microlocal Analysis}}.
\newblock Springer Verlag, 2002.

\bibitem[NS04]{Nenciu:2004ex}
Gheorghe Nenciu and Vania Sordoni.
\newblock {Semiclassical limit for multistate Klein--Gordon systems: almost
  invariant subspaces, and scattering theory}.
\newblock {\em Journal of Mathematical Physics}, 45(9):3676, 2004.

\bibitem[PST03]{PST:sapt:2002}
Gianluca Panati, Herbert Spohn, and Stefan Teufel.
\newblock {Space Adiabatic Perturbation Theory}.
\newblock {\em Adv. Theor. Math. Phys.}, 7:145--204, 2003.

\bibitem[Rob87]{Robert:tour_semiclassique:1987}
Didier Robert.
\newblock {\em {Autour de l'Approximation Semi-Classique}}.
\newblock Birkhäuser, 1987.

\bibitem[ST13]{Stiepan_Teufel:semiclassics_op_valued_symbols:2012}
Hans-Michael Stiepan and Stefan Teufel.
\newblock {Semiclassical approximations for Hamiltonians with operator-valued
  symbols}.
\newblock {\em Communications in Mathematical Physics}, 320:821--849, 2013.

\bibitem[Str57]{Stratonovich:distributions_rep_space:1957}
R.~L. Stratonovich.
\newblock {On distributions in representation space}.
\newblock {\em Sov. Phys. JETP}, 4:891, 1957.

\bibitem[Tay81]{Taylor:PsiDO:1981}
Michael~E. Taylor.
\newblock {\em Pseudodifferential Operators}.
\newblock Princeton University Press, 1981.

\bibitem[Teu03]{Teufel:adiabatic_perturbation_theory:2003}
Stefan Teufel.
\newblock {\em {Adiabatic Perturbation Theory in Quantum Dynamics}}.
\newblock Springer Verlag, 2003.

\bibitem[VGB89]{Varilly_Gracia_Bondia:Moyal_rep_spin:1989}
Joseph~C. Vàrilly and José Gracia-Bondìa.
\newblock {Moyal Representation of Spin}.
\newblock {\em Annals of Physics}, 190:107--148, 1989.

\bibitem[VGBS90]{Varilly_Gracia_Bondia_Schempp:Moyal_rep_qm:1990}
Joseph~C. Vàrilly, José Gracia-Bondìa, and Walter Schempp.
\newblock {The Moyal Representation of Quantum Mechanics and Special Function
  Theory}.
\newblock {\em Acta Applicandae Mathematica}, 11:225--250, 1990.

\bibitem[ZZ96]{Zelditch:1996tj}
Steven Zelditch and Maciej Zworski.
\newblock {Ergodicity of eigenfunctions for ergodic billiards}.
\newblock {\em Communications in Mathematical Physics}, 175(3):673--682, 1996.

\end{thebibliography}

\end{document}